\documentclass[12pt, draftclsnofoot, journal, letter, onecolumn]{IEEEtran}
\usepackage{graphicx}
\usepackage{epsfig}
\usepackage{latexsym}
\usepackage{amsfonts}
\usepackage{here}
\usepackage{rawfonts}
\usepackage[latin1]{inputenc}
\usepackage[T1]{fontenc}
\usepackage{calc}
\usepackage{capitalgreekitalic}
\usepackage{url}
\usepackage{enumerate}
\usepackage{color}
\usepackage[tbtags]{amsmath}
\usepackage{amssymb}
\usepackage{upref}
\usepackage{epic,eepic}
\usepackage{times}
\usepackage{dsfont}
\usepackage{comment}
\usepackage{cite}

\renewcommand{\Pr}{\ensuremath{\operatorname{Pr}}}

\newtheorem{theorem}{{\bf Theorem}}

\newtheorem{lemma}{{\bf Lemma}}

\newcommand{\qed}{\nobreak \ifvmode \relax \else
  \ifdim\lastskip<1.5em \hskip-\lastskip
  \hskip1.5em plus0em minus0.5em \fi \nobreak
  \vrule height0.75em width0.5em depth0.25em\fi}

\newcounter{step}
\newlength{\totlinewidth}
  {\end{list}%
  \rule{\linewidth}{1pt}}
\newcounter{substep}

  {\end{list}}

\newlength{\aligntop}
\newlength{\alignbot}
 

\IEEEoverridecommandlockouts

\begin{document}

% paper title
%\title{Minimum Power Allocation for Space-Time Coded Cooperative Routing in Multihop Wireless Networks\vspace*{-0.3em}}
%\title{Energy-Efficient Space-Time Coded Cooperative Routing in Outage-Restricted Multihop Wireless Networks\vspace*{-0.3em}}
%\title{Outage Margin in Cognitive Multiple Access Channel\vspace*{-0.3em}}
%\title{Cognitive Multiple-Antenna Network in Outage-Restricted Primary System}
\title{Integer-Forcing Message Recovering in Interference Channels}

\author{\authorblockN{Seyed Mohammad Azimi-Abarghouyi, Mohsen Hejazi, Behrooz Makki, Masoumeh Nasiri-Kenari, \emph{Senior Member, IEEE}, and Tommy Svensson, \emph{Senior Member, IEEE}}\\
   \thanks{
  S.M. Azimi-Abarghouyi and M. Nasiri-Kenari are with Electrical Engineering Department, Sharif University of Technology, Tehran, Iran. Emails:
{\textit{azimi$\_$sm@ee.sharif.edu, mnasiri@sharif.edu}}. M. Hejazi is with Electrical and Computer Engineering Department, University of Kashan, Kashan, Iran. Email: {\textit{mhejazi@ee.sharif.edu}}.
B. Makki and T. Svensson are with Department of Signals and Systems, Chalmers University
of Technology, Gothenburg, Sweden, Emails: {\textit{behrooz.makki@chalmers.se,
tommy.svensson@chalmers.se}}. This work has been supported in part by VR research link project "Green Communication via Multi-relaying".} }

%%%%%%%%%%%%%%%%%%%%%%%%%%%%%%%%%%%%%
%\author{
%\authorblockN{Behrouz Maham$^\dag$$^\ddag$ and Are Hj{\o}rungnes$^\dag$}\vspace*{0.5em}
%\authorblockA{$^\dag$UNIK -- University Graduate Center, University of Oslo, Norway\\
%$^\ddag$Department of Electrical
%    Engineering, Stanford University, USA\\
%Email: \protect\url{bmaham@stanford.edu,
%arehj@unik.no}}\vspace*{-2.1em}
%    \thanks{This work was supported by the Research Council of Norway
%    through the project 176773/S10 entitled "Optimized Heterogeneous Multiuser MIMO Networks -- OptiMO".}%
%  }
%%%%%%%%%%%%%%%%%%%%%%%%%%%%%%%%%%%%%
% make the title area
\maketitle

\begin{abstract}
In this paper, we propose a scheme referred to as integer-forcing message recovering (IFMR) to enable receivers to recover their desirable messages in interference channels. Compared to the state-of-the-art integer-forcing linear receiver (IFLR), our proposed IFMR approach needs to decode considerably less number of messages. In our method, each receiver recovers independent linear integer combinations of the desirable messages each from two independent equations. We propose an efficient algorithm to sequentially find the equations and integer combinations with maximum rates. We evaluate the performance of our scheme and compare the results with the minimum mean-square error (MMSE) and zero-forcing (ZF), as well as the IFLR schemes. The results indicate that our IFMR scheme outperforms the MMSE and ZF schemes, in terms of achievable rate, considerably. Also, compared to IFLR, the IFMR scheme achieves slightly less rates in moderate signal-to-noise ratios, with significantly less implementation complexity.

\end{abstract}

\vspace{-10pt}
\section{Introduction}
Various wireless communication setups can be modeled as interference channels consisting of multiple coexisting transmitter-receiver pairs. To reduce the interference in such systems, there are mainly two categories of receiver structures [1]-[2]. The first category are maximum likelihood (ML)-based receivers achieving the highest possible rates [1]. However, the ML-based estimation may be practically infeasible, as the size of the search space grows exponentially with the codeword length, the number of antennas, and the number of transmitters [1]. The second category are linear receivers (LR) which have low complexity in filtering the received signals through a linear structure for decoding. LRs are often proposed based on the criteria of zero-forcing (ZF) and minimum mean-square error (MMSE) [1]-[4]. 

Recently, a novel linear receiver referred to as integer-forcing linear receiver (IFLR) has been designed to simultaneously recover the transmitted messages in point-to-point multiple-input multiple-output (MIMO) systems [5]. This idea was derived from the compute-and-forward scheme [6]. Based on noisy linear combinations of the transmitted messages, IFLR recovers independent equations of messages through a linear receiver structure. In this way, in contrast to MMSE and ZF schemes, instead of combating, IFLR exploits the interference for a higher throughput. Application of the IFLR scheme in MIMO multi-pair two-way relaying is proposed in [7]. It is shown in [8] and [9] that precoding in IFLR can achieve the full diversity and the capacity of Gaussian MIMO channels up to a gap, respectively. Also, [10] applies successive decoding in IFLR and proves its sum rate optimality.
%Algorithms are proposed in [7]-[8] to approximate the best equations. 

IFLR recovers all desirable and undesirable transmitted messages by decoding sufficient number of the best independent equations in terms of achievable rate. Hence, considering IFLR in interference networks, the complexity of the lattice decoding and also the best equation selection process grows considerably with the number of transmitters and data streams. The combination of IFLR and interference alignment [11], referred to as integer-forcing interference alignment (IFIA), is proposed in [12] to decode sufficient equations to recover the desirable messages. However, IFIA requires channel state information at the transmitter (CSIT). This is the motivation for our paper in which we design an efficient low-complexity receiver for interference channels with no need for CSIT. 

Here, we propose a linear receiver scheme, referred to as integer-forcing message recovering (IFMR), for interference networks. Benefiting from a special equation structure of IFLR, we propose a novel receiver model in which the required number of decodings is limited to twice the number of desirable messages. In our IFMR, independent integer combinations of the desirable messages are recovered in each receiver. Each integer combination, referred to as desirable combined message (DCM), is recovered by decoding two independent equations. Here, with a new formulation, the equations can be optimized such that a DCM is recovered with maximum achievable rate. Despite of its much less complexity, we prove that our sequential approach in optimizing DCMs achieves the same rate as the optimal approach when we can jointly optimize DCMs (Theorem 1). 

Instead of NP hard exhaustive search in optimizing the equations of IFMR, we present a practical and efficient suboptimal algorithm to maximize the achievable rate in polynomial time. The proposed algorithm iterates in three steps, one for the coefficient factors of the two equations and the others for the coefficient vectors of an undesirable combined message (UCM) and DCM. The associated problem with each step is solved in polynomial time. The convergence of the proposed algorithm is also proved (Theorem 3). Hence, our IFMR scheme provides a low-complexity scheme in recovering the desirable messages through a few decodings of near-optimal integer combinations in interference channels.  

Our scheme is different and much less complex compared to the IFLR scheme that uses a large number of equations for message recovery. Particularly, the complexity of IFMR does not depend on the number of transmitters and the data streams of the interfering transmitters. Also, as opposed to IFIA, our scheme requires no CSIT. 

We evaluate the performance of our scheme and compare the results with the minimum mean-square error (MMSE) and zero-forcing (ZF), as well as the IFLR schemes. The results indicate that, in all signal-to-noise ratios (SNRs), our IFMR scheme outperforms the MMSE and ZF schemes, in terms of achievable rate, substantially. Also, the IFMR scheme achieves slightly less rates in moderate SNRs, compared to IFLR, with significantly less implementation complexity. In addition, our proposed algorithm provides a tight lower bound for the results obtained via the NP hard exhaustive search. For instance, consider a three-pair interference channel with single antenna at the transmitters/receivers. Then, the achievable rate of the exhaustive search is only 1 dB better than our proposed algorithm in 1 bit/channel use. 

The remainder of this paper is organized as follows. In Section II, the system model and IFLR are briefly described. Section III presents the IFMR scheme. Numerical results are given in Section IV. Finally, Section V concludes this paper.

\textbf{Notations:} The operators ${(\mathbf{A})^*}$, $\text{det}(\mathbf{A})$, $\text{Tr}(\mathbf{A})$, $||\mathbf{A}||$, and $\text{span}\left\{\mathbf{A}\right\}$ stand for conjugate transpose, determinant, trace, frobenius norm, and the space spanned by the column vectors of matrix $\mathbf{A}$, respectively. The $\mathbf{Z}^{n \times 1}$ and $\mathbf{R}^{n \times 1}$ are the $n$ dimensional integer field and $n$ dimensional real field, respectively. Moreover, ${\rm{lo}}{{\rm{g}}^ + }\left( x \right)$ denotes ${\rm{max}}\left\{ {\log \left( x \right),0} \right\}$. The operator $\succcurlyeq$ refers to the generalized inequality associated with the positive semidefinite cone. Also, ${\nabla _{{\mathbf{a}}}}{f}$ represents the partial derivative of function $f$ with respect to vector $\mathbf{a}$. Finally, $\mathbf{I}$ and $\mathbf{1}$ stand for the identity matrix and the vector with all elements equal to one, respectively.

%\vspace{-10pt}

\section{System Model and Integer-Forcing Linear Receiver (IFLR)} 	
\subsection{System Model} We consider $K$-pair interference channels where $K$ transmitters are transmitting independent data streams to $K$ receivers simultaneously, as shown in Fig 1. It is assumed that there is no coordination among the transmitters and receivers. We assume no CSIT and, as a result, we do not use beamforming. This is an acceptable assumption in simple setups with no coordinations and central processing units in which channel state information (CSI) feedback and beamforming is infeasible. Incorporating partial CSIT is left for future work. In this system, the $k$-th transmitter and receiver are equipped with $N_{{t_k}}$ and $N_{{r_k}}$ antennas, respectively. The matrix $\mathbf{H}_{kj}$ denotes the channel matrix from transmitter $k$ to receiver $j$, with dimension ${N_{{r_j}}} \times {N_{{t_k}}}$. The elements of $\mathbf{H}_{kj}$ are assumed to be independent identically distributed (IID) Gaussian variables with variance $\rho _{kj}^2$. We focus on real-valued channels. However, our scheme and results are directly applicable to complex-valued channels via a real-valued decomposition, as in [5]-[6]. Transmitter $k$ exploits a lattice encoder with power constraint $P$ to map $N_{{t_k}}$ message streams $\mathbf{w}_k$ to a real-valued codeword matrix $\mathbf{x}_k$ with dimension ${N_{{t_k}}} \times n$, where $n$ is the codeword length. 

According to Fig. 1, the received signal at receiver $k$ is given by
\begin{eqnarray}
{\mathbf{Y}_k} = {\mathbf{H}_{kk}}{\mathbf{x}_k} +\mathop \sum \limits_{j = 1,j \ne k}^K {\mathbf{H}_{jk}}{\mathbf{x}_j} + {\mathbf{n}_k},
\end{eqnarray}
where $\mathbf{n}_k$ is IID additive white Gaussian noise with the variance $\sigma^2, \forall k$. 
\subsection{Integer-Forcing Linear Receiver (IFLR)} Since the objective of our proposed approach is to limit the complexity of the IFLR scheme [5] for interference channels, it is interesting to briefly review this scheme as follows. The readers familiar with the IFLR scheme can skip this part.
\begin{figure}[tb!]
\centering

\includegraphics[width =5in]{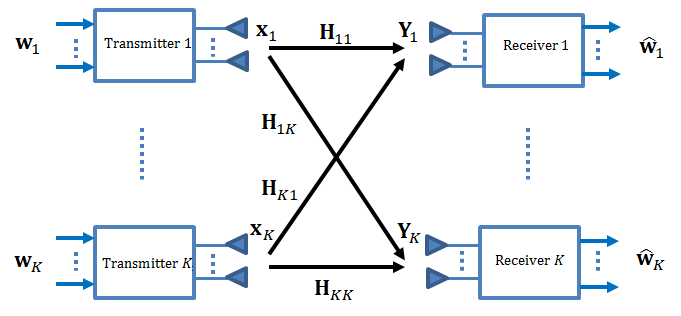}

\caption{$K$-pair interference channel.}

\end{figure}

Let us rewrite (1) as
\begin{eqnarray}
{\mathbf{Y}_k} = {\mathbf{\hat H}_k}\mathbf{X} + {\mathbf{n}_k},
\end{eqnarray}
where ${\mathbf{\hat H}_k} \buildrel \Delta \over = \left[ {{\mathbf{H}_{1k}}, \ldots ,{\mathbf{H}_{Kk}}} \right]$ and $\mathbf{X} \buildrel \Delta \over = \left[ {\begin{array}{*{20}{c}}
{{\mathbf{x}_1^*}},
 \ldots,
{{\mathbf{x}_K^*}}
\end{array}} \right]^*$. Since $\mathbf{X}$ is of size $L \buildrel \Delta \over = \mathop \sum \limits_{k = 1}^K {N_{{t_k}}}$, the IFLR scheme recovers $L$ independent equations from $\mathbf{Y}_k$. The $L$ independent equations with equation coefficient vectors (ECVs) $\mathbf{a}_l^k$, $l=1, \ldots ,L$, totally shown by matrix $\mathbf{A}^k \buildrel \Delta \over = \left[ {\begin{array}{*{20}{c}}
{\mathbf{a}_1^k},
 \ldots, 
{\mathbf{a}_L^k}
\end{array}} \right]^*$, are then solved to recover the desirable messages of the receiver $k$. 
Quantizing $\mathbf{Y}_k$, the answer of the equation with ECV $\mathbf{a}_l^k$ can be recovered as [6, Eq. (68)]
\begin{eqnarray}
{t_l^k} = \mathbf{a}{_l^k}^*\mathbf{X} = Q\left( {\mathbf{b}{_l^k}^*{\mathbf{Y}_k}} \right),
\end{eqnarray}
where $Q(\cdot)$ denotes lattice equation quantizer, and vector $\mathbf{b}_l^k$, of length $N_{{r_k}}$, is the projection vector. Also, $\mathbf{b}_l^k$ is given by [5, Eq. (28)]
\begin{eqnarray}
\mathbf{b}{_l^k}^* = \mathbf{a}{_l^k}^* \mathbf{\hat H}_k^*{\left( {\frac{1}{{\text{SNR}}}\mathbf{I} + {{\mathbf{\hat H}}_k} \mathbf{\hat H}_k^*} \right)^{ - 1}},
\end{eqnarray}
where $\text{SNR} = \frac{P}{{{\sigma ^2}}}$. Finally, the rate of the equation with ECV $\mathbf{a}_l^k$ is obtained by [5, Eq. (30)]
\begin{eqnarray}
R\left( {{\mathbf{a}_l^k}} \right) = {\log ^ + }\left( {{{\left( {\mathbf{a}{_l^k}^*\left( {\mathbf{I} - \mathbf{\hat H}_k^*{{\left( {\frac{1}{{\text{SNR}}}\mathbf{I} + {{\mathbf{\hat H}}_k} \mathbf{\hat H}_k^*} \right)}^{ - 1}}{{\mathbf{\hat H}}_k}} \right){\mathbf{a}_l^k}} \right)}^{ - 1}}} \right).
\end{eqnarray}
Hence, the optimal value of $\mathbf{A}^k$, in terms of (5), is obtained by solving the following problem
\begin{eqnarray}
\mathbf{A}_\text{opt}^k= \arg \mathop {\min}\limits_{\mathbf{A}^k \in {\mathbf{Z}^{L \times L}}} {\max _{l = 1, \ldots ,L}} \mathbf{a}{_l^k}^*\left( {\mathbf{I} -  \mathbf{\hat H}_k^*{{\left( {\frac{1}{{\text{SNR}}}\mathbf{I} + {{\mathbf{\hat H}}_k} \mathbf{\hat H}_k^*} \right)}^{ - 1}}{{ \mathbf{\hat H}}_k}} \right){\mathbf{a}_l^k},\nonumber\\
\text{subject to} \hspace{+10pt}
\text{det}(\mathbf{A}^k) \ne 0. \hspace{+200pt}
\end{eqnarray}
The problem (6) is an NP hard integer programming and its complexity grows with $L$ significantly. 

Note that the IFLR scheme does the lattice equation quantization (3) $L$ times, which increases the implementation complexity with $L$ significantly. Hence, the IFLR scheme leads to significantly higher complexity compared to the MMSE and ZF schemes [1]-[2], i.e., $L_c \buildrel \Delta \over = L - N_{{t_k}}$ more decoding for each receiver $k$.

In Section III, we propose our IFMR scheme where, independently of $K$ and ${N_{{t_i}}},\forall i \ne k$, each receiver $k$ only requires lattice equation decoding twice the number of the desirable messages, i.e., $2 \times {N_{{t_k}}}$, with a low complexity best equation selection process.

\section{Integer-Forcing Message Recovering (IFMR)}
In summary, our proposed IFMR scheme is based on the following procedure. From the received signals ${\mathbf{Y}_k}$ in (1), independent DCMs are recovered. For each DCM, the observed interfered signal is integer-forced to an UCM. Then, two independent equations of the DCM and UCM are decoded by the lattice quantizer as in (3) which lead to recovering the DCM. Finally, solving the recovered DCMs results in the desirable messages. 

In Subsection III.A, the structure of an equation in IFMR is proposed, and accordingly its receiver model is presented. Then, in Subsection III.B, we develop a sequential three-step algorithm to efficiently find the coefficient factors of the required equations in the first step and their associated coefficient vectors of UCMs and DCMs in the second and third steps, respectively, with maximum rates in polynomial time. Theorem 1 proves that our scheme with sequential selection of DCMs achieves the same rate as the optimal scheme jointly selecting DCMs. Theorem 2 proves that Lenstra-Lenstra-Lovasz (LLL) algorithm [13] is qualified to be used for the optimization problem of the first step, and Theorem 3 proves the convergence of the proposed algorithm. Simulation results are presented in Section IV where we compare the performance of our proposed scheme with those in the literature.

\subsection{Receiver Structure} We consider an equation in the general form $t^k=d^k x_k^{\text{DCM}}+e^k x_k^{\text{UCM}}$ for receiver $k$, which is an integer combination of two messages $x_k^{\text{DCM}}$ and $x_k^{\text{UCM}}$. Here, $x_k^{{\rm{DCM}}} \buildrel \Delta \over = {\mathbf{a}^k}^*{\mathbf{x}_k}$ and $x_k^{{\rm{UCM}}} \buildrel \Delta \over = \mathop \sum \limits_{j = 1,j \ne k}^K \mathbf{c}{{_j^k}^*}{\mathbf{x}_j}$ are referred to as DCM and UCM, respectively. In other words, according to the IFLR receiver structure, $t^k$ has ECV equal to $\left[ {\begin{array}{*{20}{c}}
{{d^k}{\mathbf{a}^k}}\\
{{e^k}{\mathbf{c}^k}}
\end{array}} \right]$, where ${\mathbf{c}^k} \buildrel \Delta \over = {\left[ {\mathbf{c}{{_1^k}^*}, \ldots ,\mathbf{c}{{_{k - 1}^k}^*},\mathbf{c}{{_{k + 1}^k}^*}, \ldots ,\mathbf{c}{{_K^k}^*}} \right]^*}$. $d^k$ and $e^k$ are integer coefficient factors in $\mathbf{Z}$ space. Also, $\mathbf{a}^k$ and
$\mathbf{c}_j^k, \forall j,$ are integer coefficient vectors in ${\mathbf{Z}^{{N_{{t_k}}} \times 1}}$ and ${\mathbf{Z}^{{N_{{t_j}}} \times 1}}$, respectively. 

It is straightforward to show that two equations with independent set of coefficient factors $(d_1^k,e_1^k)$ and $(d_2^k,e_2^k)$, and same $\mathbf{a}^k$ and $\mathbf{c}^k$ for the combined messages can obtain $x_k^{{\rm{DCM}}} = {\mathbf{a}^k}^*{\mathbf{x}_k}$. According to (5) and for given coefficient vector of $x_k^{\text{UCM}}$ and coefficient factors of the two equations, the rate of recovering $x_k^{\text{DCM}}$ is obtained by
\begin{eqnarray}
{R_{\text{DCM}}}\left({\mathbf{a}^k}|{\mathbf{c}^k},d_1^k,e_1^k,d_2^k,e_2^k\right) = {\rm{min}}\left\{ {R\left( {\left[ {\begin{array}{*{20}{c}}
{d_1^k{\mathbf{a}^k}}\\
{e_1^k{\mathbf{c}^k}}
\end{array}} \right]} \right),R\left( {\left[ {\begin{array}{*{20}{c}}
{d_2^k{\mathbf{a}^k}}\\
{e_2^k{\mathbf{c}^k}}
\end{array}} \right]} \right)} \right\},
\end{eqnarray}
with $R(\cdot)$ given in (5). Hence, the unconditional achievable rate of $x_k^{\text{DCM}}$ is determined by
\begin{eqnarray}
{R_{\text{DCM}}}\left({\mathbf{a}^k}\right) = \mathop {\max }\limits_{d_1^k,e_1^k,d_2^k,e_2^k\in \mathbf{Z},{\mathbf{c}^k} \in {\mathbf{Z}^{{L_c} \times 1}}} {R_{\text{DCM}}}\left({\mathbf{a}^k}|{\mathbf{c}^k},d_1^k,e_1^k,d_2^k,e_2^k\right).
\end{eqnarray}

Due to the size of $\mathbf{x}_k$, it is sufficient to recover $N_{{t_k}}$ independent DCMs. An illustration of the receiver structure is given in Fig. 2. 
\begin{figure}[tb!]
\centering

\includegraphics[width =5in]{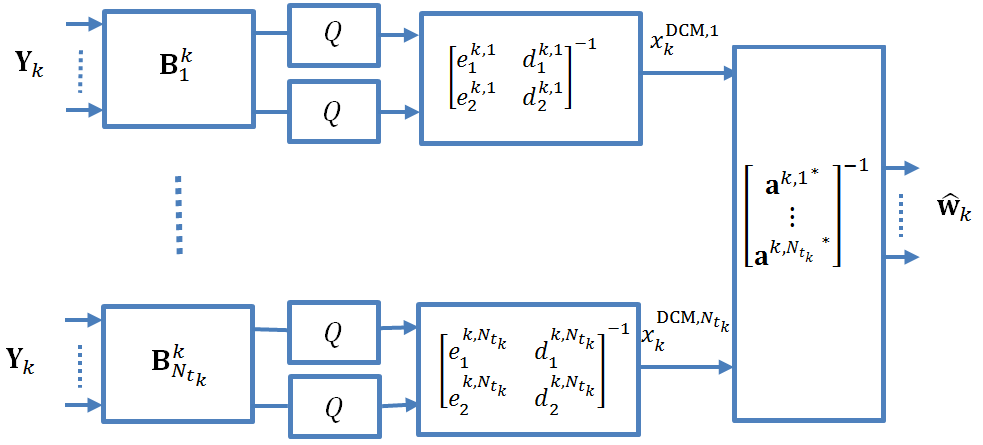}

\caption{The proposed structure of receiver $k$. In each branch $i=1,...,N_{t_k}$, $\mathbf{B}_i^k$ includes the projection vectors in (4) related to the two equations of the branch with integer coefficients $e_1^{k,i}$, $d_1^{k,i}$, $e_2^{k,i}$, $d_2^{k,i}$, $\mathbf{a}^{k,i}$, and $\mathbf{c}^{k,i}$. Note $\mathbf{c}^{k,i}$, related to undesirable recovered messages, is not shown in the figure.}

\end{figure}
\subsection{Best Integer Coefficients Selection} From (7) and (8), it is clear that the coefficients of the optimal independent DCMs with maximum rates are jointly selected from the following optimization
\begin{eqnarray}
{\max _{d{_l^{k,m}},e{_l^{k,m}}\in \mathbf{Z},{\mathbf{a}^{k,m}} \in {\mathbf{Z}^{{N_{{t_k}}} \times 1}},{\mathbf{c}^{k,m}} \in {\mathbf{Z}^{{L_c} \times 1}}}}\mathop {\min }\limits_{m = 1, \ldots ,{N_{t_k}}} \mathop {\min }\limits_{l = 1,2} R\left( {\left[ {\begin{array}{*{20}{c}}
{d{_l^{k,m}}\mathbf{a}^{k,m}}\\
{e{_l^{k,m}}\mathbf{c}^{k,m}}
\end{array}} \right]} \right),\nonumber
\end{eqnarray}
subject to
\begin{eqnarray}
\left\{ {\begin{array}{*{20}{c}}
{\text{det}\left(\left[ {\begin{array}{*{20}{c}}
{d{_1^{k,m}}}&{e{_1^{k,m}}}\\
{d{_2^{k,m}}}&{e{_2^{k,m}}}
\end{array}} \right]\right) \ne 0}, \forall m=1,...,N_{t_k}\\
{\text{det}\left( {\left[ {\mathbf{a}^{k,1}, \ldots ,\mathbf{a}^{k,{N_{t_k}}}} \right]} \right) \ne 0}
\end{array}} \right. . 
\end{eqnarray}
The problem (9) is complex, because it requires searches over space $\mathbf{Z}^{(L+2)^{N_{t_k}}\times 1}$. For this reason, we propose a sequential selection in $N_{{t_k}}$ stages which only requires a search over space $\mathbf{Z}^{{N_{t_k}}(L+2)\times 1}$. Hence, the sequential scheme is of interest because it simplifies the search process, compared to (9), significantly. In the sequential selection, each stage $t$ is to recover the best DCM $x_k^{\text{DCM},t}$ with maximum rate independently of the previously recovered messages $x_k^{\text{DCM},j}, \forall j < t$. To be more specific, in each stage $t$, it is required to solve
\begin{eqnarray}
{\max _{d_l^k,e_l^k \in \mathbf{Z},{\mathbf{a}^k} \in {\mathbf{Z}^{{N_{{t_k}}} \times 1}},{\mathbf{c}^k} \in {\mathbf{Z}^{{L_c} \times 1}}}}}\min _{l=1,2} {R\left( {\left[ {\begin{array}{*{20}{c}}
{d_l^k{\mathbf{a}^k}}\\
{e_l^k{\mathbf{c}^k}}
\end{array}} \right]} \right),\nonumber
\end{eqnarray}
subject to
\begin{eqnarray}
\left\{ {\begin{array}{*{20}{c}}
{\text{det}\left(\left[ {\begin{array}{*{20}{c}}
{d_1^k}&{e_1^k}\\
{d_2^k}&{e_2^k}
\end{array}} \right]\right) \ne 0}\\
{\text{det}\left( {\left[ {{\mathbf{a}^k},\mathbf{g}_1^k, \ldots ,\mathbf{g}_{t - 1}^k} \right]} \right) \ne 0}
\end{array}} \right. , 
\end{eqnarray}
where $\mathbf{g}_j^k$ is the integer coefficient vector associated with $x_k^{{\text{DCM},j}}$ obtained in the stage $j<t$. 
In Theorem 1, we prove that the sequential selection (10) is optimal, in the sense that it achieves the same rate as optimal search (9), with considerably less implementation complexity.
\begin{theorem}
The sequential selection (10) achieves the same rate as the optimal selection (9).
\end{theorem}
\begin{proof}
See Appendix I.
\end{proof}
Note that (10) is still an NP hard integer programming problem, requiring an exhaustive search over integer coefficients. For this reason, we propose a suboptimal scheme presented in Algorithm 1 to efficiently solve (10) in polynomial time and iteratively in three steps. In words, the algorithm is based on the following procedure. In Step I, the coefficient factors of the equations are optimized to maximize the rate of recovering given DCM and UCM. Then, in Step II, using equation factors obtained in Step I and given coefficient vector of DCM, we find the optimal coefficient vector of UCM. Finally, in Step III, for the obtained coefficient vector of UCM in Step II and the equation factors obtained in Step I, the coefficient vector of DCM is optimized. The convergence of the algorithm is proved in Theorem 3. 

\textit{Step I:} For given $\mathbf{c}^k$ and $\mathbf{a}^k$, solve 
\begin{eqnarray}
&\mathop {\min }\limits_{d_1^k,e_1^k,d_2^k,e_2^k \in \mathbf{Z}} \mathop {\max }\limits_{l = 1,2} f_l (\mathbf{a}^k,\mathbf{c}^k),\nonumber\\
&\text{subject to}\hspace{+10pt}
{\rm{det}}\left(\left[ {\begin{array}{*{20}{c}}
{d_1^k}&{e_1^k}\\
{d_2^k}&{e_2^k}
\end{array}} \right]\right) \ne 0 .
\end{eqnarray}
Defining $f_l (\mathbf{a}^k,\mathbf{c}^k)\buildrel \Delta \over = R\left( {\left[ {\begin{array}{*{20}{c}}
{d_l^k{\mathbf{a}^k}}\\
{e_l^k{\mathbf{c}^k}}
\end{array}} \right]} \right)$ and ${\mathbf{H}_k} \buildrel \Delta \over = \left[ {{\mathbf{H}_{1k}}, \ldots ,{\mathbf{H}_{\left( {k - 1} \right)k}},{\mathbf{H}_{\left( {k + 1} \right)k}}, \ldots ,{\mathbf{H}_{Kk}}} \right]$, we use (5) to expand $f_l (\mathbf{a}^k,\mathbf{c}^k)$ as
\begin{multline}
f_l (\mathbf{a}^k,\mathbf{c}^k)=d{{_l^k}^2}{\mathbf{a}^k}^*{\mathbf{a}^k} + e{{_l^k}^2}{\mathbf{c}^k}^*{\mathbf{c}^k} - \left( {d_l^k{\mathbf{a}^k}^*\mathbf{H}_{kk}^* + e_l^k{\mathbf{c}^k}^*\mathbf{H}_k^*} \right) {\left( {\frac{1}{{\text{SNR}}}\mathbf{I} + {\mathbf{H}_{kk}}\mathbf{H}_{kk}^* + {\mathbf{H}_k}\mathbf{H}_k^*} \right)^{ - 1}}\times \\ \left( {d_l^k{\mathbf{H}_{kk}}{\mathbf{a}^k} + e_l^k{\mathbf{H}_k}{\mathbf{c}^k}} \right).
\end{multline}
Hence, with some simplifications, the optimization (11) can be written as
\begin{eqnarray}
&\mathop {\min }\limits_{d_1^k,e_1^k,d_2^k,e_2^k \in \mathbf{Z}} \mathop {\max }\limits_{l = 1,2} \left[ \begin{array}{*{20}{c}}
{d_l^k}&{e_l^k}
\end{array} \right]\mathbf{U}\left[ \begin{array}{*{20}{c}}
{d_l^k}\\
{e_l^k}
\end{array} \right],\nonumber\\
&\text{subject to} \hspace{+10pt}
{\rm{det}}\left(\left[ {\begin{array}{*{20}{c}}
{d_1^k}&{e_1^k}\\
{d_2^k}&{e_2^k}
\end{array}} \right]\right) \ne 0 ,
\end{eqnarray}
where 
\begin{multline}
\mathbf{U} \buildrel \Delta \over = \left[ {\begin{array}{*{20}{c}}
{{\mathbf{a}^k}^*{\mathbf{a}^k}}&0\\
0&{{\mathbf{c}^k}^*{\mathbf{c}^k}}
\end{array}} \right] - \left[ {\begin{array}{*{20}{c}}
{{\mathbf{a}^k}^*\mathbf{H}_{kk}^*}\\
{{\mathbf{c}^k}^*\mathbf{H}_k^*}
\end{array}} \right]{\left( {\frac{1}{{\text{SNR}}}{\mathbf{I}} + {\mathbf{H}_{kk}}\mathbf{H}_{kk}^* + {\mathbf{H}_k}\mathbf{H}_k^*} \right)^{ - 1}}\left[ {\begin{array}{*{20}{c}}
{{\mathbf{H}_{kk}}{\mathbf{a}^k}}&{{\mathbf{H}_k}{\mathbf{c}^k}}
\end{array}} \right].
\end{multline}
In Theorem 2 we prove $\mathbf{U}$ to be positive definite, the proof of which uses Lemma 1 as follows.
\begin{lemma}
Matrix ${\mathbf{I}} - \frac{{\mathbf{x}{\mathbf{x}^*}}}{{{\mathbf{x}^*}\mathbf{x}}}$ is semi-definite, where $\mathbf{x} \ne 0$ is a vector in $\mathbf{R}^{L\times 1}$.
\end{lemma}
\begin{proof}
See Appendix II.
\end{proof}
\begin{theorem}
$\mathbf{U}$ is a positive definite matrix.
\end{theorem}
\begin{proof}
See Appendix III.
\end{proof}
According to Theorem 2, $\mathbf{U}$ admits a unique Cholesky decomposition. Hence, (13) can be solved efficiently in polynomial time with the LLL method [13]. 
\begin{small}
\begin{table}[t]%[t]
    %\tiny
  \centering
  \caption{%\mycaption{%\vspace*{-1em}
    "Algorithm 1"}
    \vspace*{-1em}
  \begin{tabular}{l}
      \hline
\underline{For} $t = 1, \ldots ,{N_t}$
\\
\hspace*{+10pt} \underline{Initialize} ${\mathbf{a}^{k,0}}$ and ${\mathbf{c}^{k,0}}$
\\
\hspace*{+10pt} \underline{Iterate}
\\
\hspace*{+20pt}1. Step I: Update $\left[ {\begin{array}{*{20}{c}}
{d_1^{k,j + 1}}&{e_1^{k,j + 1}}\\
{d_2^{k,j + 1}}&{e_2^{k,j + 1}}
\end{array}} \right]$ by solving (13) with the \\\hspace*{+22pt}assumption of given $\mathbf{a}^{k,j}$ and $\mathbf{c}^{k,j}$ .
\\
\hspace*{+20pt}2. Step II: Update $\mathbf{c}^{k,j+1}$ by solving (21) with the assumption of \\\hspace*{+22pt}given $\mathbf{a}^{k,j}$ and $\left[ {\begin{array}{*{20}{c}}
{d_1^{k,j + 1}}&{e_1^{k,j + 1}}\\
{d_2^{k,j + 1}}&{e_2^{k,j + 1}}
\end{array}} \right]$. 
\\

\hspace*{+20pt}3. Step III: Check which of the cases 1-3 are valid and update \\\hspace*{+22pt}$\mathbf{a}^{k,j+1}$ accordingly with the assumption of given $\mathbf{c}^{k,j+1}$ and\\\hspace*{+22pt}$\left[ {\begin{array}{*{20}{c}}
{d_1^{k,j + 1}}&{e_1^{k,j + 1}}\\
{d_2^{k,j + 1}}&{e_2^{k,j + 1}}
\end{array}} \right]$.
\\
\hspace*{+10pt}\underline{Until} 
$\min \left\{ {R\left( {\left[ {\begin{array}{*{20}{c}}
{d_1^k{\mathbf{a}^k}}\\
{e_1^k{\mathbf{c}^k}}
\end{array}} \right]} \right),R\left( {\left[ {\begin{array}{*{20}{c}}
{d_2^k{\mathbf{a}^k}}\\
{e_2^k{\mathbf{c}^k}}
\end{array}} \right]} \right)} \right\}$ converges with \\ \hspace*{+22pt}convergence threshold $\delta$. 
\\
\hspace*{+10pt}\underline{Then} $\mathbf{g}_t^k \leftarrow {\mathbf{a}^{k,{\rm{end}}}}$
\\
\underline{End}
    \vspace*{.5em}
\\

      \hline
    \end{tabular}

\end{table}
\end{small}

\textit{Step II:} For given $\mathbf{a}^k$ and coefficient factors $(e_1^k, d_1^k)$ and $(e_2^k, d_2^k)$, solve
\begin{eqnarray}
\mathop {\min }\limits_{{\mathbf{c}^k} \in \mathbf{Z}^{{L_{{c}}} \times 1}} \mathop {\max }\limits_{l = 1,2} f_l (\mathbf{a}^k,\mathbf{c}^k),
\end{eqnarray}
which, according to (12), can be rewritten by
\begin{eqnarray}
\mathop {\min }\limits_{{\mathbf{c}^k} \in \mathbf{Z}^{{L_{{c}}} \times 1}} \mathop {\max }\limits_{l = 1,2}{\mathbf{c}^k}^*\mathbf{Q}_l{\mathbf{c}^k}-2\mathbf{q}_l^*\mathbf{c}^k,
\end{eqnarray}
where 
\begin{eqnarray}
\mathbf{Q}_l \buildrel \Delta \over= e{{_l^k}^2}\mathbf{I} -e{{_l^k}^2}\mathbf{H}_k^* {\left( {\frac{1}{{\text{SNR}}}\mathbf{I} + {\mathbf{H}_{kk}}\mathbf{H}_{kk}^* + {\mathbf{H}_k}\mathbf{H}_k^*} \right)^{ - 1}}{\mathbf{H}_k},
\end{eqnarray}
\begin{eqnarray}
\mathbf{q}_l \buildrel \Delta \over= e{_l^k} d{_l^k}\mathbf{H}_{k}^*{\left( {\frac{1}{{\text{SNR}}}\mathbf{I} + {\mathbf{H}_{kk}}\mathbf{H}_{kk}^* + {\mathbf{H}_k}\mathbf{H}_k^*} \right)^{ - 1}}\mathbf{H}_{kk}\mathbf{a}_{k}.
\end{eqnarray}
The min-max quadratic problem (16) is an NP-hard integer programming. Therefore, we propose an efficient suboptimal solution for (16) which is obtained in polynomial time as follows.
First, we relax the constraint ${\mathbf{c}^k} \in {\mathbf{Z}^{{L_{c}} \times 1}}$, and let the optimal value of $\mathbf{c}^k$ to be continuous, i.e., ${\mathbf{c}^k} \in {\mathbf{R}^{{L_{c}} \times 1}}$, with the constraint $\text{diag}\left\{{\mathbf{c}^k}\right\}({\mathbf{c}^k}-\mathbf{1}) \geq \mathbf{0}$. Then, the obtained real-valued solution is rounded to its closest integer point. It is shown in [14] that the constraint $x_i(x_i-1)\geq 0$ for each element $i$ of a real-valued vector $\mathbf{x}$ can achieve a tight lower bound on the optimal value of the integer quadratic minimization of $\mathbf{x}$. 

Following the same approach as in [14], we relax (16) as
\begin{eqnarray}
&\mathop {\min }\limits_{{\mathbf{c}^k} \in \mathbf{R}^{{L_{{c}}} \times 1}} \mathop {\max }\limits_{l = 1,2}{\mathbf{c}^k}^*\mathbf{Q}_l{\mathbf{c}^k}-2\mathbf{q}_l^*\mathbf{c}^k, \nonumber\\
&\text{subject to} \hspace{10pt}
\text{diag}\left\{{\mathbf{c}^k}\right\}({\mathbf{c}^k}-\mathbf{1}) \geq \mathbf{0}.
\end{eqnarray}

With the definition of ${\mathbf{C}^k} \buildrel \Delta \over = {\mathbf{c}^k}{\mathbf{c}^k}^*$, the problem (19) is reformulated as
\begin{eqnarray}
&\mathop {\min }\limits_{{\mathbf{c}^k} \in \mathbf{R}^{{L_{{c}}} \times 1}} \mathop {\max }\limits_{l = 1,2}{\rm{\text{Tr}}}\left\{ {{\mathbf{Q}_l}{\mathbf{C}^k}} \right\} - 2\mathbf{q}_l^*{\mathbf{c}^k},\nonumber\\
&\text{subject to} \hspace{10pt}\left\{ {\begin{array}{*{20}{c}}
\text{diag}\left\{ {{\mathbf{C}^k}} \right\} \ge {\mathbf{c}^k}\\
{\mathbf{C}^k} = {\mathbf{c}^k}{\mathbf{c}^k}^*
\end{array}} \right. .
\end{eqnarray}
Then, relaxing the nonconvex constraint ${\mathbf{C}^k} = {\mathbf{c}^k}{\mathbf{c}^k}^*$ into a convex constraint ${\mathbf{C}^k} \succcurlyeq {\mathbf{c}^k}{\mathbf{c}^k}^*$, the non-convex problem (20) with the help of a Schur complement is relaxed to a convex problem as
\begin{eqnarray}
&\mathop {\min }\limits_{\varepsilon, {\mathbf{c}^k} \in \mathbf{R}^{{L_{{c}}} \times 1}} \varepsilon , \nonumber\\
&\text{subject to} \hspace{10pt}\left\{ {\begin{array}{*{20}{c}}
{\begin{array}{*{20}{c}}
{\rm{\text{Tr}}}\left\{ {{\mathbf{Q}_1}{\mathbf{C}^k}} \right\} - 2\mathbf{q}_1^*{\mathbf{c}^k} \le \varepsilon \\
{\rm{\text{Tr}}}\left\{ {{\mathbf{Q}_2}{\mathbf{C}^k}} \right\} - 2\mathbf{q}_2^*{\mathbf{c}^k} \le \varepsilon 
\end{array}}\\
{\begin{array}{*{20}{c}}
\text{diag}\left\{ {{\mathbf{C}^k}} \right\} \ge {\mathbf{c}^k}\\
\left[ {\begin{array}{*{20}{c}}
{{\mathbf{C}^k}}&{{\mathbf{c}^k}}\\
{{\mathbf{c}^k}^*}&1
\end{array}} \right] \succcurlyeq \mathbf{0}
\end{array}}
\end{array}} \right. .
\end{eqnarray}
The problem (21) is a semidefinite programming (SDP) and can be efficiently solved by CVX [15], which is a software package developed for convex optimization problems.

\textit{Step III:} For given $\mathbf{c}^k$ and coefficient factors $(e_1^k, d_1^k)$ and $(e_2^k, d_2^k)$, solve
\begin{eqnarray}
&\mathop {\min }\limits_{{\mathbf{a}^k} \in \mathbf{Z}^{{N_{{t_k}}} \times 1}} \mathop {\max }\limits_{l = 1,2} f_l (\mathbf{a}^k,\mathbf{c}^k),\nonumber\\
&\text{subject to} \hspace{10pt}
{\rm{det}}\left( {\left[ {{\mathbf{a}^k},\mathbf{g}_1^k, \ldots ,\mathbf{g}_{t - 1}^k} \right]} \right) \ne 0, 
\end{eqnarray}
which is an NP-hard integer programming. Here, because of the constraint ${\rm{det}}\left( {\left[ {{\mathbf{a}^k},\mathbf{g}_1^k, \ldots ,\mathbf{g}_{t - 1}^k} \right]} \right) \ne 0$, we cannot achieve a tight bound with an approach similar to Step II. Therefore, we propose a search over integer space $\mathbf{Z}^{{N_{{t_k}}}\times 1}$ which can obtain an efficient suboptimal solution of (22) as follows. First, we optimize (22) with a relaxation on the constraint ${\mathbf{a}^k} \in {\mathbf{Z}^{{N_{{t_k}}} \times 1}}$ as ${\mathbf{a}^k} \in {\mathbf{R}^{{N_{{t_k}}} \times 1}}$. Then, we search over a $N_{t_k}$-dimensional quantization sphere which has the obtained real valued solution ${\mathbf{a}^k} \in {\mathbf{R}^{{N_{{t_k}}} \times 1}}$ as its center, and find the best candidate according to (22). Since $f_l (\mathbf{a}^k,\mathbf{c}^k)$ is a convex quadratic function, the proposed search can achieve a tight suboptimal solution of (22) when the quantization sphere has sufficiently large radius. The quantization scheme will be further discussed in the sequel.
%so that $\max \limits_{l = 1,2} f_l (\mathbf{a}^k,\mathbf{c}^k)$ can be reduced during iterations

Here, the problem (22) is relaxed as
\begin{eqnarray}
\mathop {\min}\limits_{{\mathbf{a}^k} \in {\mathbf{R}^{{N_{t_k}} \times 1}}} \mathop {\max }\limits_{l = 1,2} f_l (\mathbf{a}^k,\mathbf{c}^{k}).
\end{eqnarray}
To obtain the solution of (23) in closed form, we use the same procedure as in [16, Subsection III.A] to convert (23) to an equivalent problem as 
\begin{eqnarray}
\mathop {{\rm{max}}}\limits_{0 \le \alpha  \le 1} \mathop {{\rm{min}}}\limits_{{\mathbf{a}^k} \in {\mathbf{R}^{{N_{{t_k}}} \times 1}}} V\left( {\alpha ,{\mathbf{a}^k}} \right),
\end{eqnarray}
where $V(\alpha ,{\mathbf{a}^k}) = \alpha {f_1}({\mathbf{a}^k},{\mathbf{c}^{k}}) + (1 - \alpha ){f_2}({\mathbf{a}^k},{\mathbf{c}^{k}})$, and $0 \le \alpha  \le 1$ is an auxiliary parameter. As details are given in Appendix IV, the solution of (24) is 
\begin{multline}
{\mathbf{a}^k} = u^k({\alpha ^*}){\left( {v^k({\alpha ^*})\mathbf{I} - \mathbf{H}_{kk}^*{{\left( {\frac{1}{{\text{SNR}}}\mathbf{I} + {\mathbf{H}_{kk}}\mathbf{H}_{kk}^* + {\mathbf{H}_k}\mathbf{H}_k^*} \right)}^{ - 1}}{\mathbf{H}_{kk}}} \right)^{ - 1}}\\ \times \mathbf{H}_{kk}^*{\left( {\frac{1}{{\text{SNR}}}\mathbf{I} + {\mathbf{H}_{kk}}\mathbf{H}_{kk}^* + {\mathbf{H}_k}\mathbf{H}_k^*} \right)^{ - 1}}{\mathbf{H}_k}{\mathbf{c}^{k}},
\end{multline}
%\vspace{-20pt}
%\begin{small}
%\begin{multline}
%{\mathbf{c}^k}^{\left( {{\alpha ^*}} \right)} = u^k({\alpha ^*})\\ \times{\left( {v^k({\alpha ^*})\mathbf{I} - \mathbf{H}_k^*{{\left( {\frac{1}{{\text{SNR}}}\mathbf{I} + {\mathbf{H}_{kk}}\mathbf{H}_{kk}^* + {\mathbf{H}_k}\mathbf{H}_k^*} \right)}^{ - 1}}{\mathbf{H}_k}} \right)^{ - 1}}\\ \times \mathbf{H}_k^*{\left( {\frac{1}{{\text{SNR}}}\mathbf{I} + {\mathbf{H}_{kk}}\mathbf{H}_{kk}^* + {\mathbf{H}_k}\mathbf{H}_k^*} \right)^{ - 1}}{\mathbf{H}_{kk}}{\mathbf{a}^{k,t-1}}.
%\end{multline}
%\end{small}
where $\alpha^*$ is obtained according to the considered three cases in Appendix IV. Also, functions $u^k(p)$ and $v^k(p)$ are defined as follows
\begin{eqnarray}
u^k(p) \buildrel \Delta \over = p{e_1^k}{d_1^k} + \left( {1 - p} \right){e_2^k}{d_2^k} ,\nonumber
\end{eqnarray}
\vspace{-20pt}
\begin{eqnarray}
v^k(p)\buildrel \Delta \over = pd{_1^k}^2 + \left( {1 - p} \right)d{_2^k}^2 .
\end{eqnarray}
%\vspace{-20pt}
%\begin{small}
%\begin{multline}
%{\mathbf{c}^k}^{\left( 0 \right)} = e_2^kd_2^k{\biggl( {e{{_2^k}^2}\mathbf{I} - \mathbf{H}_k^*{{\left( {\frac{1}{{\text{SNR}}}\mathbf{I} + {\mathbf{H}_{kk}}\mathbf{H}_{kk}^* + {\mathbf{H}_k}\mathbf{H}_k^*} \right)}^{ - 1}}\\ \times{\mathbf{H}_k}} \biggr)^{ - 1}} \mathbf{H}_k^*{\left( {\frac{1}{{\text{SNR}}}\mathbf{I} + {\mathbf{H}_{kk}}\mathbf{H}_{kk}^* + {\mathbf{H}_k}\mathbf{H}_k^*} \right)^{ - 1}}{\mathbf{H}_{kk}}{\mathbf{a}^{k,t-1}}.
%\end{multline}
%\end{small}

As a polynomial-time approach to search over the quantization sphere, we can consider slowest descent lines with directions of the eignevectors of the hessien of the cost function $f_l (\mathbf{a}^k,\mathbf{c}^k)$ in (12), i.e., $d{{_l^k}^2}\mathbf{I} - d{{_l^k}^2}\mathbf{H}_{kk}^*{\left( {\frac{1}{{\text{SNR}}}\mathbf{I} + {\mathbf{H}_{kk}}\mathbf{H}_{kk}^* + {\mathbf{H}_k}\mathbf{H}_k^*} \right)^{ - 1}} {\mathbf{H}_{kk}}$, which crosses the center ${\mathbf{a}^k} \in {\mathbf{R}^{{N_{{t_k}}} \times 1}}$ in (25). Then, the closest integer points to the lines and independent of $\mathbf{g}_1^k, \ldots ,\mathbf{g}_{t-1}^k$ are checked to find the best candidate. This approach is based on the slowest descent method which can efficiently search over discrete points [17]. 

Assume that the quantization radius is $R$ and the number of the slowest descent lines is $W$. It is straightforward to show that our approach needs to search over at most $W\times (2R+1)\times N_{t_k}$ integer points. Through the following lemma, we can exclude those ${\mathbf{a}^k}$ from the quantization sphere for which the rate (7) are zero. It also determines the maximum required radius for the quantization sphere, which guarantees to include the optimal solution of (22). The Lemma is of interest because it reduces the complexity for searching in the quantization sphere.
\begin{lemma}
Assume $e_1$, $e_2$, $d_1$, $d_2$, and $\mathbf{c}^k$ are given. The search space $\mathbf{a}^k$ with the following norm leads to rate $0$ in (7).
\begin{eqnarray}
{\left| {\left| {{\mathbf{a}^k}} \right|} \right|^2} \geq \mathop {\min }\limits_{l = 1,2} \frac{1}{{e_l^2}}\left( {1 + {\rm{\text{SNR}}}{\lambda _{{\rm{max}}}^2}\left( {{{\mathbf{\hat H}}_k}} \right) - d_l^2{\left| {\left| {{\mathbf{c}^k}} \right|} \right|^2}} \right),
\end{eqnarray}
where ${\lambda _{\rm{max}}}\left( {{\mathbf{\hat H}}_k} \right)$ is the maximum singular value of $\mathbf{\hat H}_k$.
\end{lemma}
\begin{proof}
See Appendix V.
\end{proof}

\vspace{+50pt}
Algorithm 1, summarized in Table 1, is iterated until a convergence threshold $\delta$, considered by the algorithm designer, is reached. In the simulation results, we will show the performance of our polynomial time suboptimal algorithm in comparison with the NP-hard optimal exhaustive search of the equations and UCM and DCM coefficients over the cost function of (10). The following theorem proves the convergence of Algorithm 1.

\begin{theorem}
Algorithm 1 is convergent.
\end{theorem}
\begin{proof}
See Appendix VI.
\end{proof}

\section{Simulation Results}
In this section, we provide simulation results that demonstrate the performance of the proposed IFMR scheme. Consider a three pair interference channel in which each node is equipped with $N$ antennas, unless otherwise stated. The elements of the channel matrices are assumed to have Gaussian distribution with variance 1, i.e., $\rho _{kj}^2 = 1$, $\forall k, j$. The additive white Gaussian noise has $\sigma^2 = 1$. The convergence threshold parameter $\delta$ in Algorithm 1 is set to $10^{-3}$. We average over 10000 randomly generated channel realizations.

In Figs. 3 and 4, we evaluate the achievable rates of our proposed IFMR scheme and compare the results with the state-of-the-art works, i.e., MMSE and ZF [2], and IFLR scheme [5], for $N=1$ and $N=2$, respectively. As observed, Algorithm 1 can achieve almost the same performance as in the optimal exhaustive  search-based scheme. For instance, in the cases with $N = 1$ and $N = 2$, the performance degradation, compared to the optimal exhaustive search-based approach, is less than 1 dB in 1 bit/channel use and 1 dB in 2 bit/channel use, respectively. It is also observed that the IFMR scheme outperforms the conventional MMSE and ZF receivers at all SNRs, and the performance gap increases with SNR which is because of the increase in interference. Also, the IFMR scheme achieves slightly higher rates compared to the IFLR scheme at low SNRs. It is due to the fact that the optimal equations recovered from (6) may have zero elements with high probability at low SNRs [18], whereby a subset of the equations would be enough for recovering the desirable messages. Note that the IFLR scheme leads to better achievable rates compared to the IFMR scheme at high SNRs, at the expense of much higher complexity. For example, the IFLR scheme has 2 dB improvement compared to the IFMR scheme at 1.15 bit/channel use in the one antenna case (Fig. 3) and 2.5 dB improvement at 2.5 bit/channel use in two antennas case (Fig. 4). That is because, in comparison with IFMR, the IFLR scheme has more flexibility in decoding the interference as equations. 
%In addition, the performance gap is larger for the case with $N = 1$ in comparison to $N =2$. It is because the MMSE and ZF schemes have lower degrees of freedom in design of their receivers, especially for the case with $N = 1$. 

In Fig. 5, we investigate the average number of required iterations as a function of SNR for the cases with $N = 2$. It is observed that for all considered SNRs less than 5 iterations are required for the algorithm convergence. Thus, our algorithm can be effectively applied in delay-constrained applications.

Fig. 6 shows the throughput versus the target rate $R_\text{t}$ for the case with $N = 2$. The throughput is defined as, e.g., [19, Eq. (4)]
\begin{eqnarray} 
\eta = R_\text{t}\times\left(1-\Pr(R_\text{achievable}<R_\text{t})\right).  \nonumber
\end{eqnarray}
As observed, for small values of $R_\text{t}$, the throughput increases with the rate almost linearly, because with high probability the data is correctly decoded. On the other hand, the outage probability increases and the throughput goes to zero for large values of $R_\text{t}$. Moreover, depending on the SNR, there may be a finite optimum for the target rate in terms of throughput.

In Figs. 7 and 8, the effect of the number of receiving antennas $N$ is assessed on the achievable rate and the outage probability of the proposed algorithm when each transmitter has one antenna. The outage probability is defined as $\Pr(R_\text{achievable}<R_\text{t})$. Here, $R_\text{t}$ = 1 bit/channel use is considered. As can be observed from Fig. 7, the achievable rate increases with the number of antennas $N$. For example, in sum rate of 2 bit/channel use, the system with $N = 4$ improves the power efficiency by 4 dB and 10 dB compared to the cases with $N = 3$ and $N = 2$, respectively. Also, from Fig. 8, the IFMR scheme results in diversity, i.e., the slope of the outage probability curves at high SNRs, approximately equal to $N$.

%In Fig. 5 and 6, the effect of the number of reciever antennas $N$ is assessed on the achievable rate and outage probability of the proposed algorithm when each transmitter has only two messages for transmission. Here, it is assumed that each transmitter applies a precoding matrix with all entries equal to one to meet the extra antennas.
 
\begin{figure}[t!]
\centering
\center
\vspace{-1ex}
\includegraphics[width =4.5in]{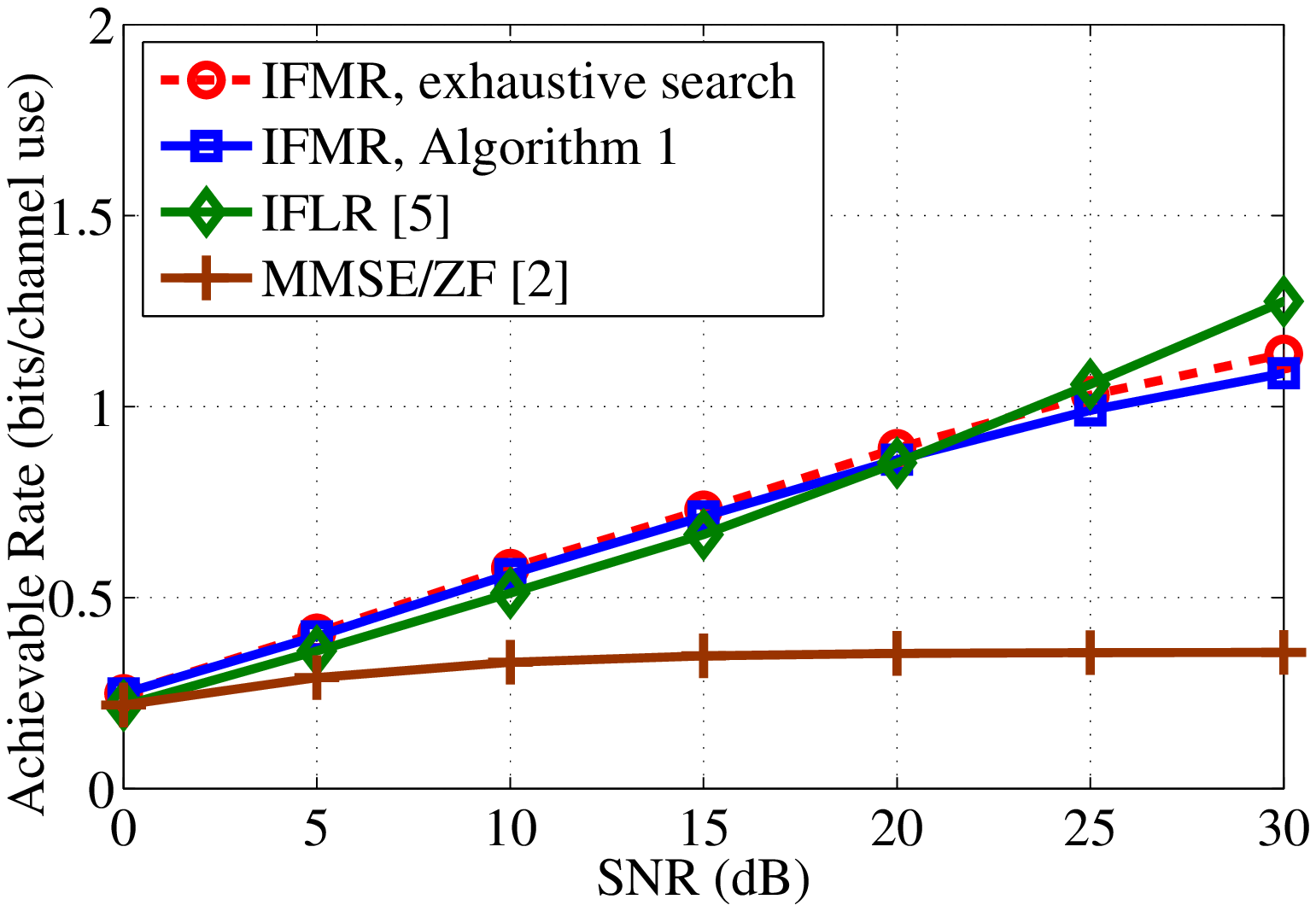}
\caption{Achievable rate of IFMR vs conventional MMSE, ZF, and IFLR for SISO, i.e., $1\times1$ MIMO, three pair interference channel.}
\end{figure}

\begin{figure}[t!]
\centering
\center
\vspace{-1ex}
\includegraphics[width =4.5in]{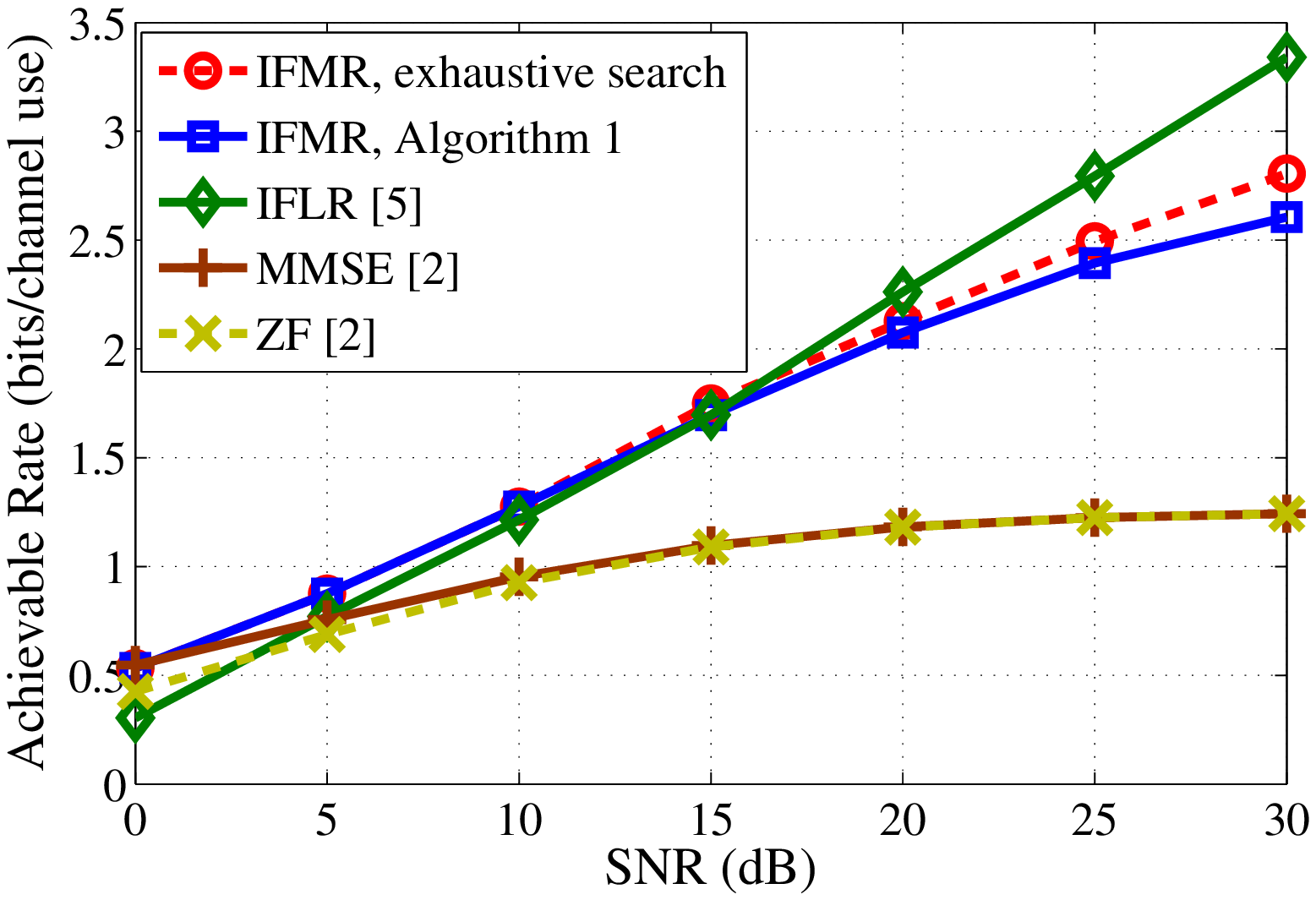}
\caption{Achievable rate of IFMR vs conventional MMSE, ZF, and IFLR for $2\times2$ MIMO three pair interference channel.}

\end{figure}

\begin{figure}[t!]
\centering
\center
\vspace{-1ex}
\includegraphics[width =4.5in]{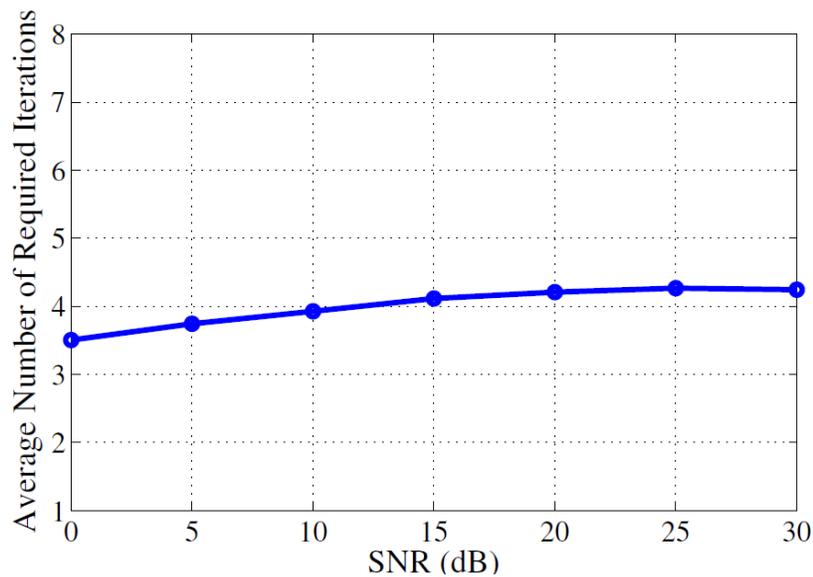}
\caption{The average number of required iterations of IFMR for $2\times2$ MIMO three pair interference channel.}

\end{figure}
\begin{figure}[t!]
\centering
\center
\vspace{-1ex}
\includegraphics[width =4.5in]{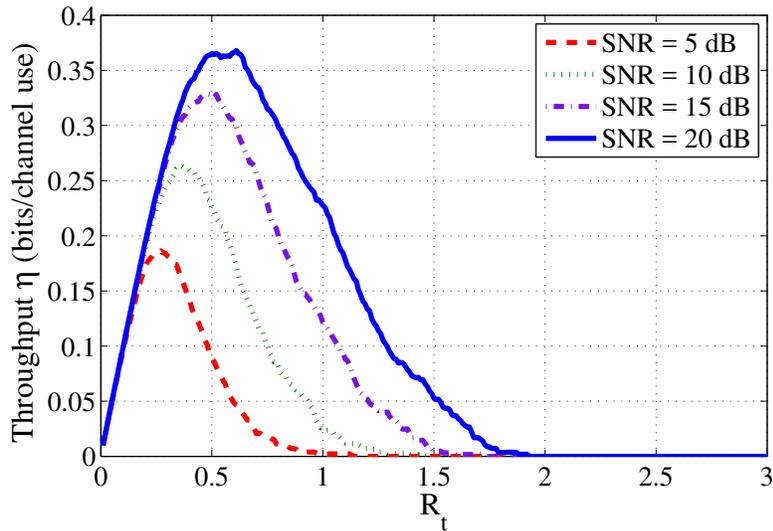}
\caption{Throughput versus the target rate $R_\text{t}$ for $2\times2$ MIMO three pair interference channel.}

\end{figure}
\begin{figure}[t!]
\centering
\center
\vspace{-1ex}
\includegraphics[width =4.5in]{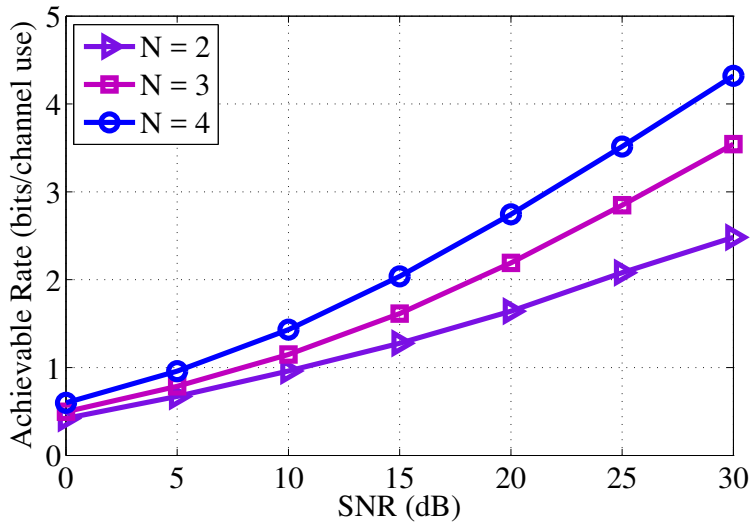}
\caption{Achievable rate of IFMR for SIMO, i.e., $1\times N$ MIMO, three pair interference channel with $N$ receiving antennas.}

\end{figure}

\begin{figure}[t!]
\centering
\center
\vspace{-1ex}
\includegraphics[width =4.5in]{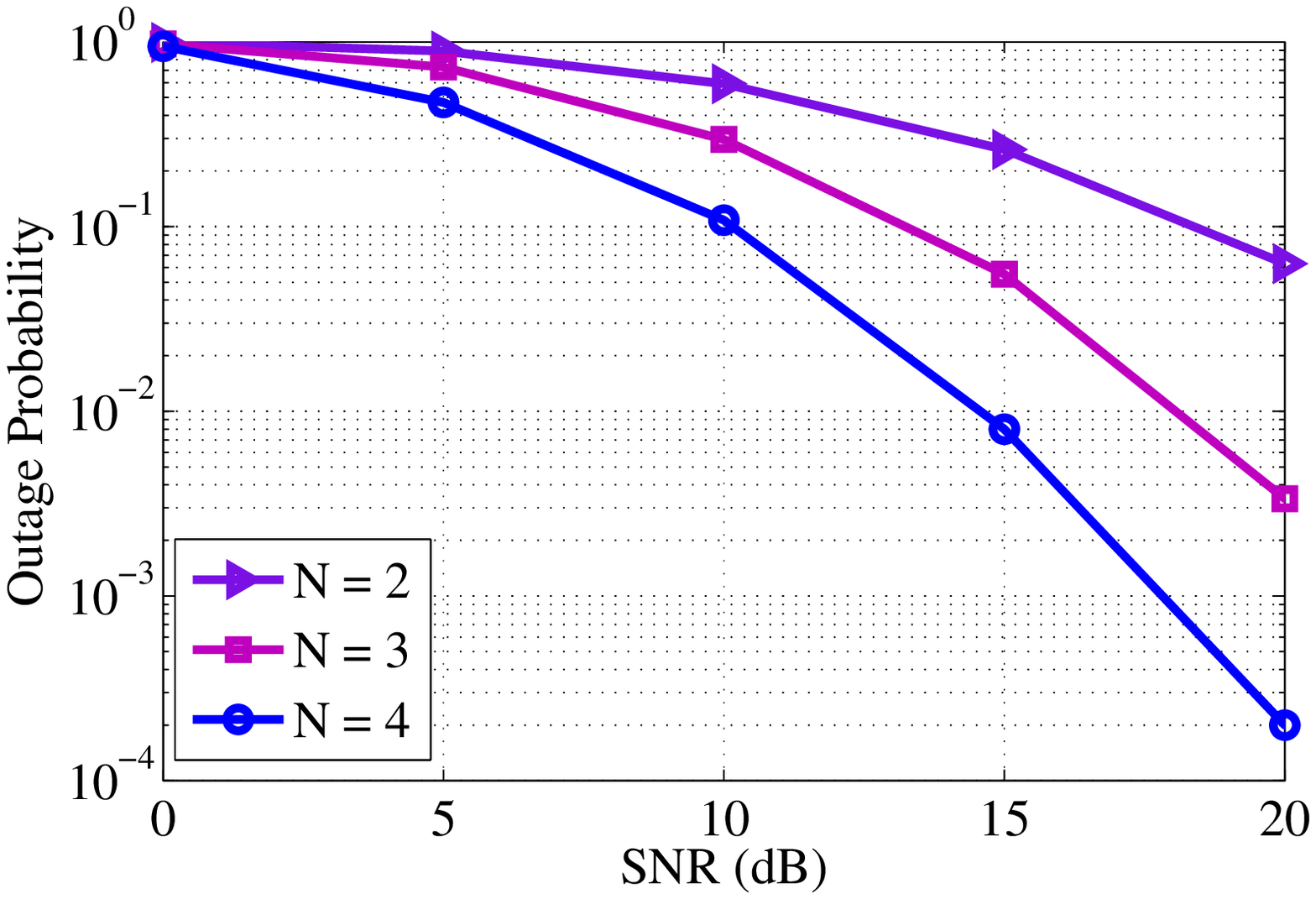}
\caption{Outage probability of IFMR for SIMO, i.e., $1\times N$ MIMO, three pair interference channel with $N$ receiving antennas and $R_\text{t}$ = 1 bit/channel use.}

\end{figure}

%\vspace{-10pt}
\section{Conclusion}
In this paper, we proposed a low-complexity linear receiver scheme, referred to as IFMR, for interference channels. In IFMR, an integer combination of the desirable messages of each receiver can be recovered with the help of only two equations independently of the number of transmitters and data streams. We first proved that the sequential selection of the integer combinations can achieve the same rate as in the optimally joint selection. Then, we proposed a suboptimal algorithm to optimize the required equations and integer combinations in polynomial time and proved its convergence. Despite of its much less complexity for IFMR, our proposed algorithm can achieve almost the same performance as in the exhaustive search scheme. The IFMR scheme also shows a significantly better performance, in terms of the achievable rate, in comparison with the MMSE and ZF schemes.
%\vfill\eject
\appendices
\section{Proof of Theorem 1}
Let the independent DCM coefficient vectors $\left\{ {{\mathbf{g}_1^k}, \ldots ,{\mathbf{g}_{N_t}^k}} \right\}$ be selected by the sequential method in (10). According to the constraint in (10), we have $R_{\text{DCM}}\left( {{\mathbf{g}_1^k}} \right) \ge R_{\text{DCM}}\left( {{\mathbf{g}_2^k}} \right) \ge  \ldots  \ge R_{\text{DCM}}\left( {{\mathbf{g}_{N_t}^k}} \right)$. Hence, the achievable rate of the sequential technique is ${R_{\text{seq}}} = {N_{{t}}} \times {R_{\text{DCM}}}\left( {{\mathbf{g}_{N_{{t}}}^k}} \right)$. Suppose also that the independent set $\left\{ {\mathbf{d}_1^k}, \ldots ,{\mathbf{d}_{N_t}^k} \right\}$, i.e., $\text{rank}\left\{ {\mathbf{d}_1^k}, \ldots ,{\mathbf{d}_{N_t}^k} \right\} = {N_t}$, are the optimum solution of (9). Without loss of generality, assume that $R_{\text{DCM}}\left( {\mathbf{d}_1^k} \right) \ge R_{\text{DCM}}\left( {\mathbf{d}_2^k} \right) \ge  \ldots  \ge R_{\text{DCM}}\left( {\mathbf{d}_{N_t}^k} \right)$. Thus, the achievable rate of the optimal technique is ${R_{\text{opt}}} = {N_{t}} \times {R_{\text{DCM}}}\left( {\mathbf{d}_{N_{t}}^k} \right)$. 

Using contradiction, assume $R_{\text{opt}} > R_{\text{seq}}$. Hence, ${R_{\text{DCM}}}\left( {{\mathbf{d}_{N_{t}}^k}} \right) > {R_{\text{DCM}}}\left( {{\mathbf{g}_{N_{t}}^k}} \right)$. From (10), $\mathbf{g}_{N_{{t}}}^k$ is obtained from two equations which have the maximum achievable rate among all set of two equations whose associated DCM coefficient vectors are linearly independent of $\left\{ {\mathbf{g}_1^k}, \ldots ,{\mathbf{g}_{{N_t}-1}^k} \right\}$. This implies that every DCM coefficient vector with a rate higher than ${R_{\text{DCM}}}\left( \mathbf{g}_{N_{t}}^k \right)$ is linearly dependent to the set $\left\{ {{\mathbf{g}_1^k}, \ldots ,{\mathbf{g}_{{N_t}-1}^k}} \right\}$. Thus, we conclude $\mathbf{d}_{N_{t}}^k$ exists in the $\text{span}\left\{ {\mathbf{g}_1^k}, \ldots ,{\mathbf{g}_{{N_t}-1}^k} \right\}$. As a result, for all ${\mathbf{d}_{{i}}^k}, \forall i \le N_{t}$, we have 
%has the maximum rate among all DCM vectors that are linearly independent of
\begin{eqnarray}
\left\{ {{\mathbf{d}_1^k}, \ldots ,{\mathbf{d}_{{N_t}}^k}} \right\} \in \text{span}\left\{ {{\mathbf{g}_1^k}, \ldots ,{\mathbf{g}_{{N_t}-1}^k}} \right\},
\end{eqnarray}
which indicates that $\text{rank}\left\{ {{\mathbf{d}_1^k}, \ldots ,{\mathbf{d}_{{N_t}}^k}} \right\} \le {N_t}-1$. However, this contradicts the assumption of linear-independency of these equations. Hence, ${R_{\text{opt}}} = {R_{\text{seq}}}$.

\section{Proof of Lemma 1}
For every vector $\mathbf{y} \ne 0$ in $\mathbf{R}^{L\times 1}$, we can write
\begin{eqnarray}
{\bf{\Gamma }} \buildrel \Delta \over = {\mathbf{y}^*}\left( {{\mathbf{I}} - \frac{{\mathbf{x}{\mathbf{x}^*}}}{{{\mathbf{x}^*}\mathbf{x}}}} \right)\mathbf{y} = {\mathbf{y}^*}\mathbf{y} - {\mathbf{y}^*}\frac{{\mathbf{x}{\mathbf{x}^*}}}{{{\mathbf{x}^*}\mathbf{x}}}\mathbf{y} = {\mathbf{y}^*}\mathbf{y} - \frac{1}{{{\mathbf{x}^*}\mathbf{x}}}{\mathbf{y}^*}\mathbf{x}{\mathbf{x}^*}\mathbf{y}.
\end{eqnarray}
Then, from the Cauchy-Schwarz inequality ${\mathbf{y}^*}\mathbf{x}{\mathbf{x}^*}\mathbf{y} \le \left( {{\mathbf{y}^*}\mathbf{y}} \right)\left( {{\mathbf{x}^*}\mathbf{x}} \right)$, we conclude ${\bf{\Gamma }} \ge 0$. Thus, ${\mathbf{I}} - \frac{{\mathbf{x}{\mathbf{x}^*}}}{{{\mathbf{x}^*}\mathbf{x}}}$ is semi-definite.

\section{Proof of Theorem 2}
From the definition of $\mathbf{U}$ in (14) and adding then subtracting a term, we can write
\begin{multline}
\mathbf{U} = \left[ {\begin{array}{*{20}{c}}
{{\mathbf{a}^k}^*{\mathbf{a}^k}}&0\\
0&{{\mathbf{c}^k}^*{\mathbf{c}^k}}
\end{array}} \right] - \left[ {\begin{array}{*{20}{c}}
{{\mathbf{a}^k}^*\mathbf{H}_{kk}^*}\\
{{\mathbf{c}^k}^*\mathbf{H}_k^*}
\end{array}} \right]{\left( {\frac{1}{{\text{SNR}}}{\mathbf{I}} + {\mathbf{H}_{kk}}\frac{{{\mathbf{a}^k}{\mathbf{a}^k}^*}}{{{\mathbf{a}^k}^*{\mathbf{a}^k}}}\mathbf{H}_{kk}^* + {\mathbf{H}_k}\frac{{{\mathbf{c}^k}{\mathbf{c}^k}^*}}{{{\mathbf{c}^k}^*{\mathbf{c}^k}}}\mathbf{H}_k^*} \right)^{ - 1}}\\ \times \left[ {\begin{array}{*{20}{c}}
{{\mathbf{H}_{kk}}{\mathbf{a}^k}}&{{\mathbf{H}_k}{\mathbf{c}^k}}
\end{array}} \right]+ \left[ {\begin{array}{*{20}{c}}
{{\mathbf{a}^k}^*\mathbf{H}_{kk}^*}\\
{{\mathbf{c}^k}^*\mathbf{H}_k^*}
\end{array}} \right]{\left( {\frac{1}{{\text{SNR}}}{\mathbf{I}} + {\mathbf{H}_{kk}}\frac{{{\mathbf{a}^k}{\mathbf{a}^k}^*}}{{{\mathbf{a}^k}^*{\mathbf{a}^k}}}\mathbf{H}_{kk}^* + {\mathbf{H}_k}\frac{{{\mathbf{c}^k}{\mathbf{c}^k}^*}}{{{\mathbf{c}^k}^*{\mathbf{c}^k}}}\mathbf{H}_k^*} \right)^{ - 1}} \left[ {\begin{array}{*{20}{c}}
{{\mathbf{H}_{kk}}{\mathbf{a}^k}}&{{\mathbf{H}_k}{\mathbf{c}^k}}
\end{array}} \right] \\- \left[ {\begin{array}{*{20}{c}}
{{\mathbf{a}^k}^*\mathbf{H}_{kk}^*}\\
{{\mathbf{c}^k}^*\mathbf{H}_k^*}
\end{array}} \right]{\left( {\frac{1}{{\text{SNR}}}{\mathbf{I}} + {\mathbf{H}_{kk}}\mathbf{H}_{kk}^* + {\mathbf{H}_k}\mathbf{H}_k^*} \right)^{ - 1}}\left[ {\begin{array}{*{20}{c}}
{{\mathbf{H}_{kk}}{\mathbf{a}^k}}&{{\mathbf{H}_k}{\mathbf{c}^k}}
\end{array}} \right].
\end{multline}
%\end{multicols}
%According to matrix inverses identities $\mathbf{A}^{-1}-\mathbf{B}^{-1}=\mathbf{A}^{-1} (\mathbf{B}-\mathbf{A})\mathbf{B}^{-1}$ and $(\mathbf{A}+\mathbf{BCD})^{-1}=\mathbf{A}^{-1}-\mathbf{A}^{-1} \mathbf{B}(\mathbf{C}^{-1}+\mathbf{DA}^{-1} \mathbf{B})^{-1} \mathbf{DA}^{-1}$ for every matrixes $\mathbf{A}$, $\mathbf{B}$, $\mathbf{C}$, and $\mathbf{D}$ [11], we can rewrite (32) as
According to matrix inverses identities in [20, Eqs. (159) and (165)], we can rewrite (30) as
\begin{multline}
\mathbf{U}=\biggl( \left[ {\begin{array}{*{20}{c}}
{\frac{1}{{{\mathbf{a}^k}^*{\mathbf{a}^k}}}}&0\\
0&{\frac{1}{{{\mathbf{c}^k}^*{\mathbf{c}^k}}}}
\end{array}} \right] + \text{SNR}\left[ {\begin{array}{*{20}{c}}
{\frac{1}{{{\mathbf{a}^k}^*{\mathbf{a}^k}}}{\mathbf{a}^k}^*\mathbf{H}_{kk}^*}\\
{\frac{1}{{{\mathbf{c}^k}^*{\mathbf{c}^k}}}{\mathbf{c}^k}^*\mathbf{H}_k^*}
\end{array}} \right] \left[ {\begin{array}{*{20}{c}}
{\frac{1}{{{\mathbf{a}^k}^*{\mathbf{a}^k}}}{\mathbf{H}_{kk}}{\mathbf{a}^k}}&{\frac{1}{{{\mathbf{c}^k}^*{\mathbf{c}^k}}}{\mathbf{H}_k}{\mathbf{c}_k}}
\end{array}} \right] \biggr)^{ - 1} \\+ \left[ {\begin{array}{*{20}{c}}
{{\mathbf{a}^k}^*\mathbf{H}_{kk}^*}\\
{{\mathbf{c}^k}^*\mathbf{H}_k^*}
\end{array}} \right] \biggl\{\left( {\frac{1}{{\text{SNR}}}{\mathbf{I}} + {\mathbf{H}_{kk}}\frac{{{\mathbf{a}_k}{\mathbf{a}_k}^*}}{{{\mathbf{a}^k}^*{\mathbf{a}^k}}}\mathbf{H}_{kk}^* + {\mathbf{H}_k}\frac{{{\mathbf{c}_k}{\mathbf{c}_k}^*}}{{{\mathbf{c}^k}^*{\mathbf{c}^k}}}\mathbf{H}_k^*} \right)^{ - 1}\\ \times \left( {{\mathbf{H}_{kk}}\mathbf{H}_{kk}^* - {\mathbf{H}_{kk}}\frac{{{\mathbf{a}_k}{\mathbf{a}_k}^*}}{{{\mathbf{a}^k}^*{\mathbf{a}^k}}}\mathbf{H}_{kk}^* + {\mathbf{H}_k}\mathbf{H}_k^* - {\mathbf{H}_k}\frac{{{\mathbf{c}_k}{\mathbf{c}_k}^*}}{{{\mathbf{c}^k}^*{\mathbf{c}^k}}}\mathbf{H}_k^*} \right)
 \left( {\frac{1}{{\text{SNR}}}{\mathbf{I}} + {\mathbf{H}_{kk}}\mathbf{H}_{kk}^* + {\mathbf{H}_k}\mathbf{H}_k^*} \right)^{ - 1}\biggr\} \\ \times \left[ {\begin{array}{*{20}{c}}
{{\mathbf{H}_{kk}}{\mathbf{a}^k}}&{{\mathbf{H}_k}{\mathbf{c}_k}}
\end{array}} \right].
\end{multline}
It is straightforward to show that matrices $\mathbf{F}$, $\mathbf{G}$, and $\mathbf{T}$ with
\begin{eqnarray}
\mathbf{F} = \biggl( \left[ \begin{array}{*{20}{c}}
{\frac{1}{{{\mathbf{a}^k}^*{\mathbf{a}^k}}}}&0\\
0&{\frac{1}{{{\mathbf{c}^k}^*{\mathbf{c}^k}}}}
\end{array} \right] + \text{SNR}\left[ \begin{array}{*{20}{c}}
{\frac{1}{{{\mathbf{a}^k}^*{\mathbf{a}^k}}}{\mathbf{a}^k}^*\mathbf{H}_{kk}^*}\\
{\frac{1}{{{\mathbf{c}^k}^*{\mathbf{c}^k}}}{\mathbf{c}^k}^*\mathbf{H}_k^*}
\end{array} \right] \left[ {\begin{array}{*{20}{c}}
{\frac{1}{{{\mathbf{a}^k}^*{\mathbf{a}^k}}}{\mathbf{H}_{kk}}{\mathbf{a}^k}}&{\frac{1}{{{\mathbf{c}^k}^*{\mathbf{c}^k}}}{\mathbf{H}_k}{\mathbf{c}^k}}
\end{array}} \right] \biggr)^{ - 1},\nonumber
\end{eqnarray}
\begin{eqnarray}
\mathbf{G} = \left( {\frac{1}{{\text{SNR}}}{\mathbf{I}} + {\mathbf{H}_{kk}}\frac{{{\mathbf{a}^k}{\mathbf{a}^k}^*}}{{{\mathbf{a}^k}^*{\mathbf{a}^k}}}\mathbf{H}_{kk}^* + {\mathbf{H}_k}\frac{{{\mathbf{c}^k}{\mathbf{c}^k}^*}}{{{\mathbf{c}^k}^*{\mathbf{c}^k}}}\mathbf{H}_k^*} \right)^{ - 1},\nonumber
\end{eqnarray}
\begin{eqnarray}
\mathbf{T} = \left( {\frac{1}{{\text{SNR}}}{\mathbf{I}} + {\mathbf{H}_{kk}}\mathbf{H}_{kk}^* + {\mathbf{H}_k}\mathbf{H}_k^*} \right)^{ - 1},
\end{eqnarray}
are positive definite. 

According to Lemma 1, since ${\mathbf{H}_{kk}}\mathbf{H}_{kk}^* - {\mathbf{H}_{kk}}\frac{{{\mathbf{a}^k}{\mathbf{a}^k}^*}}{{{\mathbf{a}^k}^*{\mathbf{a}^k}}}\mathbf{H}_{kk}^*$ and ${\mathbf{H}_k}\mathbf{H}_k^* - {\mathbf{H}_k}\frac{{{\mathbf{c}^k}{\mathbf{c}^k}^*}}{{{\mathbf{c}^k}^*{\mathbf{c}^k}}}\mathbf{H}_k^*$ are positive semi-definite matrices, the matrix $\mathbf{X}$ with
\begin{multline}
\mathbf{X}={\left( {\frac{1}{{\text{SNR}}}{\mathbf{I}} + {\mathbf{H}_{kk}}\frac{{{\mathbf{a}^k}{\mathbf{a}^k}^*}}{{{\mathbf{a}^k}^*{\mathbf{a}^k}}}\mathbf{H}_{kk}^* + {\mathbf{H}_k}\frac{{{\mathbf{c}^k}{\mathbf{c}^k}^*}}{{{\mathbf{c}^k}^*{\mathbf{c}^k}}}\mathbf{H}_k^*} \right)^{- 1}}\\ \times \biggl( {{\mathbf{H}_{kk}}\mathbf{H}_{kk}^* - {\mathbf{H}_{kk}}\frac{{{\mathbf{a}^k}{\mathbf{a}^k}^*}}{{{\mathbf{a}^k}^*{\mathbf{a}^k}}}\mathbf{H}_{kk}^* + {\mathbf{H}_k}\mathbf{H}_k^*- {\mathbf{H}_k}\frac{{{\mathbf{c}^k}{\mathbf{c}^k}^*}}{{{\mathbf{c}^k}^*{\mathbf{c}^k}}}\mathbf{H}_k^*} \biggr){\left( {\frac{1}{{\text{SNR}}}{\mathbf{I}} + {\mathbf{H}_{kk}}\mathbf{H}_{kk}^* + {\mathbf{H}_k}\mathbf{H}_k^*} \right)^{ - 1}},
\end{multline}
 is also semi-definite. Hence, the overall matrix $\mathbf{U}$, which is sum of a positive definite matrix and a semi-definite matrix, is positive definite.

\section{Details for the Solution of (24)}
For (24), we further define a function
\begin{eqnarray}
V\left( \alpha  \right) \buildrel \Delta \over = {\rm{mi}}{{\rm{n}}_{{\mathbf{a}^k} \in {\mathbf{R}^{{N_{{t_k}}} \times 1}}}}V\left( {\alpha ,{\mathbf{a}^k}} \right) = V\left( {\alpha ,{\mathbf{a}^k}^{\left( \alpha  \right)}} \right),
\end{eqnarray}
where ${\mathbf{a}^k}^{\left( \alpha  \right)}$ minimizes $V(\alpha,{\mathbf{a}^k})$ for given $\alpha$. 
Let ${\alpha ^*}$ denote the solution of ${\rm{ma}}{{\rm{x}}_{0 \le \alpha  \le 1}}V\left( \alpha  \right)$. There are three cases according to the relationship of ${f_1}\left( {{\mathbf{a}^k}^{\left( {{\alpha ^*}} \right)}},{\mathbf{c}^{{k}}} \right)$ and ${f_2}\left( {{\mathbf{a}^k}^{\left( {{\alpha ^*}} \right)}},{\mathbf{c}^{{k}}} \right)$, one of which includes the solution of (24).

\underline{Case 1}: If ${\alpha ^*} = 0$, we have 
\begin{eqnarray}
{f_1}\left( {{\mathbf{a}^k}^{\left( 0 \right)},{\mathbf{c}^{k}}} \right) \le {f_2}\left( {{\mathbf{a}^k}^{\left( 0 \right)},{\mathbf{c}^{k}}} \right).
\end{eqnarray}
Hence, (24) is changed to
\begin{eqnarray}
\mathop {{\rm{min}}}\limits_{{\mathbf{a}^k}^{\left( 0 \right)} \in {\mathbf{R}^{{N_{{t_k}}} \times 1}}} {f_2}\left( {{\mathbf{a}^k}^{\left( 0 \right)},{\mathbf{c}^{k}}} \right),
\end{eqnarray}
which can be effectively solved by setting the derivative of ${f_2}\left( {{\mathbf{a}^k}^{\left( 0 \right)},{\mathbf{c}^{k}}} \right)$ with respect to ${\mathbf{a}^k}^{\left( 0 \right)}$ equal to zero. Hence, according to (12), the optimal ${\mathbf{a}^k}^{\left( 0 \right)}$ is given by 
\begin{multline}
{\nabla _{{\mathbf{a}^k}}}{f_2}\left( {{\mathbf{a}^k},{\mathbf{c}^{k}}} \right) = \biggl( {d{{_2^k}^2}\mathbf{I} - \mathbf{H}_{kk}^*{{\biggl( {\frac{1}{{\text{SNR}}}\mathbf{I} + {\mathbf{H}_{kk}}\mathbf{H}_{kk}^* + {\mathbf{H}_k}\mathbf{H}_k^*} \biggr)}^{ - 1}} {\mathbf{H}_{kk}}} \biggr){\mathbf{a}^k}\\ - e_2^kd_2^k\mathbf{H}_{kk}^*{\left( {\frac{1}{{\text{SNR}}}\mathbf{I} + {\mathbf{H}_{kk}}\mathbf{H}_{kk}^* + {\mathbf{H}_k}\mathbf{H}_k^*} \right)^{ - 1}}{\mathbf{H}_k}{\mathbf{c}^{k}} = 0,
\end{multline}
%\vspace{-16pt}
%\begin{small}
%\begin{multline}
%{\nabla _{{\mathbf{c}^k}}}{f_2}\left( {{\mathbf{a}^{k, t-1}},{\mathbf{c}^{k}}} \right) = \biggl( {e{{_2^k}^2}\mathbf{I} - \mathbf{H}_k^*{{\left( {\frac{1}{{\text{SNR}}}\mathbf{I} + {\mathbf{H}_{kk}}\mathbf{H}_{kk}^* + {\mathbf{H}_k}\mathbf{H}_k^*} \right)}^{ - 1}}\\ \times {\mathbf{H}_k}} \biggr){\mathbf{c}^{k}}- e_2^kd_2^k\mathbf{H}_k^*{\left( {\frac{1}{{\text{SNR}}}\mathbf{I} + {\mathbf{H}_{kk}}\mathbf{H}_{kk}^* + {\mathbf{H}_k}\mathbf{H}_k^*} \right)^{ - 1}}{\mathbf{H}_{kk}}{\mathbf{a}^{k, t-1}} = 0,
%\end{multline}
%\end{small}
which respectively leads to
\begin{multline}
{\mathbf{a}^k}^{\left( 0 \right)} = {e_2^k}{d_2^k}\biggl( d{{_2^k}^2}\mathbf{I} - \mathbf{H}_{kk}^*{\left( {\frac{1}{{\text{SNR}}}\mathbf{I} + {\mathbf{H}_{kk}}\mathbf{H}_{kk}^* + {\mathbf{H}_k}\mathbf{H}_k^*} \right)}^{ - 1} {\mathbf{H}_{kk}} \biggr)^{ - 1} \\ \times \mathbf{H}_{kk}^*\left( \frac{1}{{\text{SNR}}}\mathbf{I} + {\mathbf{H}_{kk}}\mathbf{H}_{kk}^*+ {\mathbf{H}_k}\mathbf{H}_k^* \right)^{ - 1}{\mathbf{H}_k}{\mathbf{c}^{{k}}}.
\end{multline}
\underline{Case 2}: If ${\alpha ^*} = 1$, then we have 
\begin{eqnarray}
{f_1}\left( {{\mathbf{a}^k}^{\left( 1 \right)},{\mathbf{c}^{k}}} \right) \ge {f_2}\left( {{\mathbf{a}^k}^{\left( 1 \right)},{\mathbf{c}^{k}}} \right).
\end{eqnarray}
Thus, similar to Case 1, we can find the ${\mathbf{a}^k}^{\left( 1 \right)}$ as
\begin{multline}
{\mathbf{a}^k}^{\left( 1 \right)} = {e_1^k}{d_1^k}\biggl( d{{_1^k}^2}\mathbf{I} - \mathbf{H}_{kk}^*{{\left( {\frac{1}{{\text{SNR}}}\mathbf{I} + {\mathbf{H}_{kk}}\mathbf{H}_{kk}^* + {\mathbf{H}_k}\mathbf{H}_k^*} \right)}^{ - 1}}{\mathbf{H}_{kk}} \biggr)^{ - 1} \\ \times\mathbf{H}_{kk}^*{\left( {\frac{1}{{\text{SNR}}}\mathbf{I} + {\mathbf{H}_{kk}}\mathbf{H}_{kk}^* + {\mathbf{H}_k}\mathbf{H}_k^*} \right)^{ - 1}}{\mathbf{H}_k}{\mathbf{c}^{k}}.
\end{multline}
%\vspace{-20pt}
%\begin{small}
%\begin{multline}
%{\mathbf{c}^k}^{\left( 1 \right)} = e_1^kd_1^k{\biggl( {e{{_1^k}^2}\mathbf{I} - \mathbf{H}_k^*{{\left( {\frac{1}{{\text{SNR}}}\mathbf{I} + {\mathbf{H}_{kk}}\mathbf{H}_{kk}^* + {\mathbf{H}_k}\mathbf{H}_k^*} \right)}^{ - 1}}\\ \times{\mathbf{H}_k}} \biggr)^{ - 1}} \mathbf{H}_k^*{\left( {\frac{1}{{\text{SNR}}}\mathbf{I} + {\mathbf{H}_{kk}}\mathbf{H}_{kk}^* + {\mathbf{H}_k}\mathbf{H}_k^*} \right)^{ - 1}}{\mathbf{H}_{kk}}{\mathbf{a}^{k,t-1}}.
%\end{multline}
%\end{small}
\underline{Case 3}: If $0 < {\alpha ^*} < 1$, then we have 
\begin{eqnarray}
{f_1}\left( {{\mathbf{a}^k}^{\left( {{\alpha ^*}} \right)},{\mathbf{c}^{k}}} \right) = {f_2}\left( {{\mathbf{a}^k}^{\left( {{\alpha ^*}} \right)},{\mathbf{c}^{k}}} \right),
\end{eqnarray}
in which ${\alpha ^*}$ can be found by the Bisection method. In this case, (24) is rephrased as
\begin{eqnarray}
{\rm{mi}}{{\rm{n}}_{{\mathbf{a}^k}^{\left( {{\alpha ^*}} \right)} \in {\mathbf{R}^{{N_{{t_k}}} \times 1}}}}V\left( {{\alpha ^*},{\mathbf{a}^k}^{\left( {{\alpha ^*}} \right)}} \right) ={\alpha ^*}{f_1}\left( {{\mathbf{a}^k}^{\left( {{\alpha ^*}} \right)},{\mathbf{c}^{k}}} \right) + \left( {1 - {\alpha ^*}} \right){f_2}\left( {{\mathbf{a}^k}^{\left( {{\alpha ^*}} \right)},{\mathbf{c}^{k}}} \right),
\end{eqnarray}
which can be solved by setting the derivative of $V\left( {{\alpha ^*},{\mathbf{a}^k}^{\left( {{\alpha ^*}} \right)}} \right)$ with respect to ${\mathbf{a}^k}^{\left( {{\alpha ^*}} \right)}$ equal to zero. With the same arguments and using some manipulations, ${\mathbf{a}^k}^{\left( {{\alpha ^*}} \right)} $ is obtained by
\begin{multline}
{\mathbf{a}^k}^{\left( {{\alpha ^*}} \right)} = u^k({\alpha ^*}){\left( {v^k({\alpha ^*})\mathbf{I} - \mathbf{H}_{kk}^*{{\left( {\frac{1}{{\text{SNR}}}\mathbf{I} + {\mathbf{H}_{kk}}\mathbf{H}_{kk}^* + {\mathbf{H}_k}\mathbf{H}_k^*} \right)}^{ - 1}}{\mathbf{H}_{kk}}} \right)^{ - 1}}\\ \times \mathbf{H}_{kk}^*{\left( {\frac{1}{{\text{SNR}}}\mathbf{I} + {\mathbf{H}_{kk}}\mathbf{H}_{kk}^* + {\mathbf{H}_k}\mathbf{H}_k^*} \right)^{ - 1}}{\mathbf{H}_k}{\mathbf{c}^{k}},
\end{multline}
%\vspace{-20pt}
%\begin{small}
%\begin{multline}
%{\mathbf{c}^k}^{\left( {{\alpha ^*}} \right)} = u^k({\alpha ^*})\\ \times{\left( {v^k({\alpha ^*})\mathbf{I} - \mathbf{H}_k^*{{\left( {\frac{1}{{\text{SNR}}}\mathbf{I} + {\mathbf{H}_{kk}}\mathbf{H}_{kk}^* + {\mathbf{H}_k}\mathbf{H}_k^*} \right)}^{ - 1}}{\mathbf{H}_k}} \right)^{ - 1}}\\ \times \mathbf{H}_k^*{\left( {\frac{1}{{\text{SNR}}}\mathbf{I} + {\mathbf{H}_{kk}}\mathbf{H}_{kk}^* + {\mathbf{H}_k}\mathbf{H}_k^*} \right)^{ - 1}}{\mathbf{H}_{kk}}{\mathbf{a}^{k,t-1}}.
%\end{multline}
%\end{small}
where 
\begin{eqnarray}
u^k({\alpha ^*}) \buildrel \Delta \over = {\alpha ^*}{e_1^k}{d_1^k} + \left( {1 - {\alpha ^*}} \right){e_2^k}{d_2^k} ,\nonumber
\end{eqnarray}
\vspace{-20pt}
\begin{eqnarray}
v^k({\alpha ^*})\buildrel \Delta \over =  {\alpha ^*}d{_1^k}^2 + \left( {1 - {\alpha ^*}} \right)d{_2^k}^2 .
\end{eqnarray}

\section{Proof of Lemma 2}
An equation with ECV $\mathbf{a}_l^k$ has rate (5) equal to zero if ${\left| {\left| {{\mathbf{a}_l^k}} \right|} \right|^2} \geq 1+{\rm{\text{SNR}}}{\lambda _{{\rm{max}}}^2}\left( {{{\mathbf{\hat H}}_k}} \right)$ [5]. Thus, the rate (7) is zero if 
\begin{eqnarray}
{e_1^2\left| {\left| {{\mathbf{a}^k}} \right|} \right|^2+d_1^2\left| {\left| {{\mathbf{c}^k}} \right|} \right|^2} \geq 1+{\rm{\text{SNR}}}{\lambda _{{\rm{max}}}^2}\left( {{{\mathbf{\hat H}}_k}} \right),
\end{eqnarray}
or
\begin{eqnarray}
 {e_2^2\left| {\left| {{\mathbf{a}^k}} \right|} \right|^2+d_2^2\left| {\left| {{\mathbf{c}^k}} \right|} \right|^2} \geq 1+{\rm{\text{SNR}}}{\lambda _{{\rm{max}}}^2}\left( {{{\mathbf{\hat H}}_k}} \right).
 \end{eqnarray}
 Accordingly, we should have\\ 
${\left| {\left| {{\mathbf{a}^k}} \right|} \right|^2} \geq \frac{1}{{e_1^2}}\left( {1 + {\rm{\text{SNR}}}{\lambda _{{\rm{max}}}^2}\left( {{{\mathbf{\hat H}}_k}} \right) - d_1^2{\left| {\left| {{\mathbf{c}^k}} \right|} \right|^2}} \right)$ or $\frac{1}{{e_2^2}}\left( {1 + {\rm{\text{SNR}}}{\lambda _{{\rm{max}}}^2}\left( {{{\mathbf{\hat H}}_k}} \right) - d_2^2{\left| {\left| {{\mathbf{c}^k}} \right|} \right|^2}} \right)$, which completes the proof.
\section{Proof of Theorem 3}
For each $t$, assume $\epsilon^j (\mathbf{a}^{k,j},\mathbf{c}^{k,j}) = \mathop {{\rm{max}}}\limits_{l = 1,2} {f_l}^j\left( {{\mathbf{a}^{k,j}},{\mathbf{c}^{k,j}}} \right)$, where ${f_l}^j\left( {{\mathbf{a}^{k,j}},{\mathbf{c}^{k,j}}} \right)$ corresponds to the $j$-th iteration. For the iteration $j+1$ of Step I, we have $\epsilon^{j+1} (\mathbf{a}^{k,j},\mathbf{c}^{k,j}) \leq \epsilon^j (\mathbf{a}^{k,j},\mathbf{c}^{k,j})$, in Step II, $\epsilon^{j+1} (\mathbf{a}^{k,j},\mathbf{c}^{k,j+1} ) \leq \epsilon^{j+1} (\mathbf{a}^{k,j},\mathbf{c}^{k,j})$, and in Step III, $\epsilon^{j+1} (\mathbf{a}^{k,j+1},\mathbf{c}^{k,j+1} ) \leq \epsilon^{j+1} (\mathbf{a}^{k,j},\mathbf{c}^{k,j+1})$. According to $f_l (\mathbf{a}^k,\mathbf{c}^k)$, the latter is guaranteed when we assume the quantization sphere has sufficiently large radius to find a suitable $\mathbf{a}^{k,j+1}$. Even for a small quantization sphere with no candidate, we can update as $\mathbf{a}^{k,j+1} = \mathbf{a}^{k,j}$ which in the worst case of $\mathbf{c}^{k,j+1} = \mathbf{c}^{k,j}$ leads to $\epsilon^{j+1} (\mathbf{a}^{k,j+1},\mathbf{c}^{k,j+1} ) = \epsilon^{j+1} (\mathbf{a}^{k,j},\mathbf{c}^{k,j})$. Hence, $\epsilon^{j+1} (\mathbf{a}^{k,j+1},\mathbf{c}^{k,j+1}) \leq \epsilon^j (\mathbf{a}^{k,j},\mathbf{c}^{k,j})$ at the end of iteration $j+1$. In this way, in each iteration, the function $\epsilon = \mathop {{\rm{max}}}\limits_{l = 1,2} f_l (\mathbf{a}^k,\mathbf{c}^k )$ either decreases or remains unchanged, and is lower bounded by zero. Thus, the proposed algorithm is convergent. 

\end{document}